\documentclass[11pt]{article}
\usepackage{amsmath,amsthm,mathrsfs,bm}

\usepackage[initials,alphabetic]{amsrefs}
\usepackage[utf8]{inputenc}

\usepackage[dvipsnames]{xcolor}

\usepackage{amsfonts,amstext,amssymb,verbatim,epsfig,psfrag,stmaryrd,extarrows}
\usepackage{pgf,tikz}
\usepackage{float}
\usetikzlibrary{arrows}
\usepackage{hyperref}
\usepackage{color}
\usepackage{transparent}
\usepackage{subfigure}
\usepackage{prettyref}
\newrefformat{pr}{Proposition \ref{#1}}
\newrefformat{co}{Corollary \ref{#1}}
\newrefformat{le}{Lemma \ref{#1}}
\newrefformat{th}{Theorem \ref{#1}}
\newrefformat{ex}{Example \ref{#1}}
\newrefformat{re}{Remark \ref{#1}}
\newrefformat{se}{Section \ref{#1}}

\usepackage{esint}

\usepackage{bbm,mathrsfs}

\usepackage[normalem]{ulem}

\definecolor{forestgreen}{rgb}{0.1333,0.5451,0.1333}
\definecolor{navyblue}{rgb}{0,0,0.5}
\definecolor{darkgreen}{rgb}{0,0.3922,0}

\hypersetup{colorlinks=true, citecolor  = teal, linkcolor=MidnightBlue}
\makeatletter
\let\reftagform@=\tagform@
\def\tagform@#1{\maketag@@@{(\ignorespaces\textcolor{black}{#1}\unskip\@@italiccorr)}}
\renewcommand{\eqref}[1]{\textup{\reftagform@{\ref{#1}}}}
\makeatother

\makeatletter

\allowdisplaybreaks

\def\lessim{\ \lower4pt\hbox{$
		\buildrel{\displaystyle <}\over\sim$}\ }
\def\gessim{\ \lower4pt\hbox{$\buildrel{\displaystyle >}
		\over\sim$}\ }

\def\si{\sigma}

\newcommand{\pref}{\prettyref}

\newtheorem{lemma}{\bf Lemma}[section]

\newtheorem{theorem}[lemma]{\bf Theorem}

\newtheorem*{ackn}{\bf Acknowledgements}
\theoremstyle{remark}

\newtheorem{remark}{Remark}[section]
\newtheorem{definition}[lemma]{\bf Definition}
\newtheorem{assumption}{\bf Assumption}

\numberwithin{equation}{section}

\newcommand{\8}{\infty}

\newcommand{\px}{\mathcal{P}}

\newcommand{\rz}{\mathbb{R}}

\newcommand{\al}{\alpha}

\renewcommand{\si}{\sigma}

\newcommand{\la}{\lambda}

\newcommand{\Crt}{\mathrm{Crt}}

\newcommand{\GOE}{\mathrm{GOE}}

\newcommand{\argmin}{\mathrm{argmin}}
\newcommand{\argmax}{\mathrm{argmax}}
\newcommand{\dd}{\mathrm{d}}

\title{Hessian spectrum at the global minimum of the spherical pure-like mixed $p$-spin glasses}
\author{Hao Xu\thanks{School of Mathematics and Statistics, Central South University, xuhaokd001@gmail.com} \and Haoran Yang\thanks{School of Mathematical Sciences, Peking University, yanghr@pku.edu.cn.} 
}

\begin{document}

\maketitle

\abstract
We study the large $N$-dimensional limit of the Hessian spectrum at the global minimum of some subclasses of the spherical mixed $p$-spin models. Specifically, we show that its empirical spectral measure converges in probability to a shifted and rescaled semicircle law and does not have outliers. Our method follows the second moment approach developed recently in \cite{BSZ20}, from which the ground state energy can be derived for the $pure$-$like$ mixed $p$-spin model. By analyzing the complexity function with given radial derivative and energy, we derive the convergence of the Hessian spectrum from the vanishing mean number of critical points. For the $1$-RSB model, the ground state energy was explicitly computed in \cite{huang2023constructive}. Combined with the complexity function of local maxima with given radial derivative obtained in \cite{belius2022complexity}, this allows us to obtain the corresponding results more directly. Our result extends those corresponding results in the regime of topology trivialization.  
\section{Introduction}
The Hamiltonians of spherical spin glasses are Gaussian smooth functions defined on the 
sphere in $N$-dimensional Euclidean space. Although the covariance function of the Hamiltonian depends only on the angle between two points, its landscapes are known to be rich and complex. The most captivating aspect of spin glass theory is the Parisi formula for the free energy, first predicted by Parisi in 1979 \cite{parisi1979infinite} for the Ising spin models. After decades of progress, Talagrand \cite{talagrand2006parisi} was the first to provide a rigorous mathematical proof of the conjecture aforementioned for the Ising models with even spins, and later extended similar results to the spherical models \cite{talagrand2006free} also with even spins. This formula was further extended to general mixed $p$-spin models by \cite{pan} for the Ising case and by \cite{chen2013aizenman} for the spherical case, respectively. 
On the other hand, the mean number of critical points (or those with a given index) and topology of its level sets of the Hamiltonian of spherical spin glasses have been rigorously investigated in \cite{ABA13,ABC13} by using the Kac-Rice formula \cite[theorem 12.1.1]{AT07}. Since then, the Kac--Rice formula has been applied to study the moments of critical points under various conditions and in different settings; see \cite{subag2017complexity, BSZ20,belius2022complexity,Be22} and the references therein. More recently, Huang and Sellke \cite[Lemma 2.8]{huang2023constructive} obtained an ultrametric tree of pure states for the 1-RSB model (see Definition \ref{defi}), each with approximately the ground state energy, through a second moment calculation of the critical points based on the Kac--Rice formula. To study the law of a random field $f_N(\boldsymbol{\sigma})$ evaluated at a sample from the Gibbs measure associated to the spherical Hamiltonian, Dembo and Subag \cite[Proposition 6.1]{dembo2024disordered} derived an upper bound for the complexity of a class of rotationally invariant functions taking values in a special random set via the Kac--Rice formula. In this work, we will be interested in the behavior of Hessian spectrum at the global minimum of some
subclasses of the spherical mixed $p$-spin models (similarly at the global maximum, by symmetry).

We now introduce the spherical mixed $p$-spin  model. Let $S^{N-1}(\sqrt{N})$ be the sphere of radius $\sqrt{N}$ in the $N$-dimensional Euclidean space. The Hamiltonian of the spherical mixed $p$-spin model is given by
\begin{align}\label{model} H_N(\boldsymbol{\sigma}) =\sum_{p=2}^{\infty} \frac{\gamma_p}{N^{(p-1) / 2}} \sum_{i_1, \ldots, i_p=1}^N J_{i_1, \ldots, i_p}^{(p)} \sigma_{i_1} \cdots \sigma_{i_p}, \quad \boldsymbol{\sigma}=(\sigma_1,\dots,\sigma_N) \in S^{N-1}(\sqrt{N}),
\end{align}
where $J_{i_1, \ldots, i_p}^{(p)}$ are i.i.d.\ real standard Gaussian random variables and $\left\{\gamma_p\right\}_{p \geq 2}$ is a sequence of
 deterministic constants such that
\begin{align}\label{decay}
\sum_{p=2}^{\infty} 2^p \gamma_p^2<\infty.
\end{align}
Let us define
$$
\xi(x)=\sum_{p = 2}^{\infty} \gamma_p^2 x^p.
$$
A direct calculation yields that the covariance of the Hamiltonian satisfies
$$
\mathbb{E}\left[H_{N}(\boldsymbol{\sigma}) H_{N}\left(\boldsymbol{\sigma}^{\prime}\right)\right]=N\xi\left(R\left(\boldsymbol{\sigma}, \boldsymbol{\sigma}'\right)\right),
$$
where 
$R\left(\boldsymbol{\sigma}, \boldsymbol{\sigma}'\right):=\frac{1}{N} \sum_{i=1}^N \sigma_i \sigma_i'$ is the overlap between any two spin configurations $\boldsymbol{\sigma}$ and $\boldsymbol{\sigma}'$. We fix the variance of $H_N$ by assuming
$\xi(1)=\sum_{p=2}^{\infty} \gamma_p^2=1.$ If $\xi(x)=x^p$ for some $p\geq 3$, the model (\ref{model}) is called the pure $p$-spin model. The assumption (\ref{decay})
ensures that the infinite sum in the model (\ref{model}) converges almost surely and the Hamiltonian $H_N$ is almost surely smooth and Morse \cite[Theorem 11.3.1]{AT07}, which guarantees the application of the Kac-Rice formula. Denote the ground state energy in the spherical version of the mixed $p$-spin model by
$$
GS_N=\frac{1}{N}\min _{\boldsymbol{\sigma} \in S^{N-1}(\sqrt{N})} H_N(\boldsymbol{\sigma}),
$$ 
and the corresponding ground state by
$$
\boldsymbol{\sigma}^*=\operatornamewithlimits{\argmin}_{\boldsymbol{\sigma} \in S^{N-1}(\sqrt{N})} H_N(\boldsymbol{\sigma}).
$$ 
To study the ground state energy, the classical method is to study the free energy
\begin{align}\label{free}
F_N(\beta)=\frac{1}{\beta N} \log \int_{S^{N-1}(\sqrt{N})} e^{-\beta H_{N}(\boldsymbol{\sigma})} \Lambda_N(\mathrm{d} \boldsymbol{\sigma}),
\end{align}
where $\Lambda_N$ is the normalized surface measure on the sphere $S^{N-1}(\sqrt{N})$ and the parameter $\beta$ is called the
inverse temperature. It is well known that the limiting ground state energy exists almost surely [Lemma 6] \cite{chen}, $G S:=\lim _{N \rightarrow \infty}GS_N$, which can be deduced by sending the inverse temperature $\beta$ to infinity in the definition of the free energy (\ref{free}).
Based on the Crisanti-Sommers representation for free energy, a variational formula for the limiting ground state energy $GS$ was established in \cite{chen,2016dynamic}, which should be regarded as the analogue of the Parisi formula at zero-temperature. 
 
To understand the landscape of the Hamiltonian $H_N(\boldsymbol{\sigma})$,  
Auffinger, Ben Arous and Černý
\cite{ABC13} firstly investigated the mean number of critical points
of the spherical pure $p$-spin model. For any Borel set $B\subset\mathbb{R}$, denote by $\operatorname{Crt}_{N,k}(B)$  the number of critical points of $H_N(\boldsymbol{\sigma})$ when the value of $H_N(\boldsymbol{\sigma})/N$ is restricted to $B$ with index $k$. In this pioneering work, they proved that
$$
\begin{aligned}
\lim _{N \rightarrow \infty} \frac{1}{N} \log \left(\mathbb{E}\left\{\operatorname{Crt}_{N, k}((-\infty, u))\right\}\right)  =\Theta_{p, k}(u),
\end{aligned}
$$
where the explicit expression of $\Theta_{p, k}(u)$ is given in \cite[(2.16)]{ABC13}, which is commonly referred to as the complexity function of the pure $p$-spin model with index $k$. Setting $E_{\infty}(p)=2\sqrt{\frac{p-1}{p}}$, the function $\Theta_{p, k}(u)$ is known to be a strictly increasing function on $(-\infty,-E_{\infty}(p))$ and takes the constant
value $\frac{1}{2} \log (p-1)-\frac{p-2}{p}>0$ for $u\geq - E_{\8}(p)$. Let $-E_0(p)$ be the unique zero of the function $\Theta_{p, 0}(u)$. For any open set $B\subset\mathbb{R}$, denote by $\operatorname{Crt}_{N}(B)$ the number of critical points of $H_N(\boldsymbol{\sigma})$ when the value of $H_N(\boldsymbol{\sigma})/N$ is restricted to $B$. Additionally, the second moment of $\operatorname{Crt}_{N}((-\infty, u))$ was calculated in \cite{subag2017complexity}, which showed that the number of critical points is concentrated around their means: for any $p \geq 3$ and $u \in\left(-E_0(p),-E_{\infty}(p)\right)$,
$$
\lim _{N \rightarrow \infty} \frac{\mathbb{E}\left\{\left(\operatorname{Crt}_N((-\infty, u))\right)^2\right\}}{\left(\mathbb{E}\left\{\operatorname{Crt}_N((-\infty, u))\right\}\right)^2}=1.
$$
As a corollary, it was deduced in \cite[Appendix IV]{subag2017complexity} that
\begin{align}\label{gse}
GS=-E_0(p), \quad a.s.
\end{align}
This result was also obtained earlier in \cite[Theorem 2.12]{ABC13} for $p\geq 4$ even. The complexity function was also calculated for the spherical mixed $p$-spin model \cite{ABA13}, which states that for any open set $B\subset\mathbb{R}$,
$$
\lim _{N \rightarrow \infty} \frac{1}{N} \log \mathbb{E} \Crt_{N, k}(B)=\sup _{u \in B} \theta_{k, \xi}(u),
$$
where $\theta_{k, \xi}(u)$ is explicitly given in \cite[(2.10)]{ABA13}. Setting $E_{\infty}^{\prime}(\xi)=\frac{2 \xi^{\prime}(1) \sqrt{\xi^{\prime \prime}(1)}}{\xi^{\prime}(1)+\xi^{\prime \prime}(1)}$ as in \cite[(1.12)]{ABA13}, the function $\theta_{k, \xi}(u)$ is strictly increasing on $\left(-\infty,-E_{\infty}^{\prime}(\xi)\right)$ and strictly decreasing on $\left(-E_{\infty}^{\prime}(\xi), +\infty\right)$ with the unique maximum $\theta_{k, \xi}(-E_{\infty}^{\prime}(\xi))>0$. Let 
\begin{align}\label{Gxi}
G\left(\xi^{\prime}, \xi^{\prime \prime}\right)=\log \frac{\xi^{\prime \prime}(1)}{\xi^{\prime}(1)}-\frac{\left(\xi^{\prime \prime}(1)-\xi^{\prime}(1)\right)\left(\xi^{\prime \prime}(1)-\xi^{\prime}(1)+\xi^{\prime }(1)^2\right)}{\xi'' (1) \xi' (1)^2}.
\end{align}
A mixed model with covariance function $\xi(x)$ is called $pure$-$like$ if $G\left(\xi^{\prime}, \xi^{\prime \prime}\right)>0$,  $critical$ if $G\left(\xi^{\prime}, \xi^{\prime \prime}\right)=0$, or $full$ if
$G\left(\xi^{\prime}, \xi^{\prime \prime}\right)<0$. Denote by $-E_0(\xi)$ the smallest zero of $\theta_{0, \xi}(u)$. 
In Section 4 of their work, they claimed that $\xi(x)$ belonging to the $pure$-$like$ mixture  class is equivalent to $E_0(\xi)=f_1$, where $f_1$ can be interpreted as the zero-temperature limit of $1$-RSB Parisi functional \cite[Proposition 6]{ABA13}. Also through a second moment calculation, Ben Arous, Subag and Zeitouni \cite{BSZ20} proved that $GS=-E_0(\xi)$ for the $pure$-$like$ mixture model under Assumption \ref{assum}.

Recently, \cite{Be22} considered the spherical mixed $p$-spin model on the unit sphere $S^{N-1}$ with an external field
$$
H_N^h(\boldsymbol{\sigma})=H_N(\boldsymbol{\sigma})+N h \mathbf{u}_N \cdot \boldsymbol{\sigma} ,
$$
where $h \geq 0$ is the strength of the external field and $\mathbf{u}_N \in S^{N-1}$ is a deterministic sequence. When $h^2>\xi''(1)-\xi'(1)$, they showed that the Hamiltonian $H_N^h(\boldsymbol{\sigma})$ has trivial geometry, meaning that the only critical points of $H_N^h$ are a single maximum and a single minimum. It is clear that the property above implies the vanishing of the maximum value of the complexity function. In this setting, they also found the convergence of the largest eigenvalues of the Hessian at global maximum \cite[Theorem 1.1]{Be22}. In fact, similar arguments can be used to prove the convergence of the empirical spectral measures of the Hessian at global maximum, with the help of a large deviation principle (LDP) for the spectral measures of the GOE matrices \cite{benarous1997large}. Later, \cite{XZ22} also confirmed the trivial topology phenomenon for locally isotropic Gaussian random fields defined on $\mathbb{R}^N $ and studied the convergence of empirical spectral measures of the Hessian at global minimum provided the value of external field is large, which is previously known in the physics literature \cite{FLD18}. 

Out of the topology trivialization regime, it may be difficult to find the limiting Hessian spectrum at global minimum, since the maximum value of complexity function can be positive. Subag \cite{subag2021following}[Lemma 3] deduced that uniformly for all $\boldsymbol{\sigma}\in S^{N-1}(\sqrt{N})$, the Hessian spectrum at $\boldsymbol{\sigma}$ can be approximated by a semicircle measure with radius $2\sqrt{\xi''(1)}$ and shifted by the rescaled radial derivative $\partial_r H_N(\boldsymbol{\sigma})$, so the key to identifying the Hessian spectrum at global minimum is to determine the radial derivative at global minimum. However, it is known that $GS=-E_0(p)$ \cite[Appendix IV]{subag2017complexity} (see also \cite[Proposition 3]{chen}) for the pure $p$-spin model without external fields ($h$=0) and $GS=-E_0(\xi)$ \cite[Theorem 1.1]{BSZ20} for the $pure$-$like$ mixture model also without external fields (under Assumption \ref{assum}). This means that the smallest zero of the complexity function exactly matches the corresponding $GS$. By analyzing the complexity function $F(x,y)$ defined in (\ref{fxy}), we can prove the convergence of the Hessian spectrum at global minimum $\boldsymbol{\sigma}^*$ in the thermodynamic limit $N \rightarrow \infty$ for these subclasses models, which is the main focus of this paper. Moreover, for the 1-RSB model, the value of $GS$ can be computed explicitly. Given $x_*=-GS$ in the complexity $F(x,y)$ aforementioned, the maximum value of $F(-x_*,y)$ over variable $y$ is zero, and is obtained uniquely at $-y_*$ defined in (\ref{x0y0}), which should represent the corresponding value of the radial derivative at global minimum. Thus, we can also identify the Hessian spectrum at global minimum for this model. For more general models, it is unknown whether this property still holds and will require a more detailed investigation in the future. Instead, the $full$ mixture model contradicts this property (see \cite[Corollary 4.1]{ABA13}), which may pose a significant obstacle to solving this problem. 

\subsection{Notations and main results}
Let $\GOE_N$ be an $N\times N$ matrix from the Gaussian Orthogonal Ensemble (GOE), i.e., $\GOE_N$ is a real symmetric matrix whose entries $(\GOE_N)_{i j}, i \leq j$ are independent centered Gaussian random variables with variance
$$
\mathbb{E} (\GOE_N)_{i j}^{2}=\frac{1+\delta_{i j}}{N} .
$$
We denote by $\mathcal{P}(\mathbb{R})$ the space of probability measures on $\rz$, equipped with the bounded Lipschitz metric for two probability measures $\mu, \nu\in \px(\rz)$,
\begin{align}\label{eq:measd}
d(\mu,\nu)=\sup \left\{ \left\lvert\int f \dd \mu-\int f \dd\nu \right\rvert \colon \|f\|_\8\le1, \|f\|_L\le 1 \right\},
\end{align}
where $\|f\|_\8$ and $\|f\|_L$ denote the $L^\8$ norm and Lipschitz constant of the function $f$, respectively. The empirical spectral measure of an $N\times N$ matrix $M$ is denoted by $L_{N}(M)=\frac{1}{N} \sum_{i=1}^{N} \delta_{\lambda_{i}(M)}$, where $\la_i(M)$ are the eigenvalues of $M$. We also use $\lambda_{\min}(M)$ and $\lambda_{\max}(M)$ to denote the smallest eigenvalue and the largest eigenvalue of the matrix $M$, respectively.
Let $\si_{c,r}$ denote the semicircle measure with center $c$ and radius $r$, which has the following density with respect to the Lebesgue measure
\begin{align*}
    \si_{c,r}(\dd x) = \frac{2}{\pi r^2}\sqrt{r^2-(x-c)^2} \, \dd x.
\end{align*}
In particular, we denote the standard semicircle measure by $\si_{\rm sc} = \si_{0,2}$. We endow the sphere $S^{N-1}(\sqrt{N})$ with the standard Riemannian structure, induced by the Euclidean Riemannian metric on $\mathbb{R}^N$. We fix an orthonormal frame field $E=\left(E_i\right)_{i=1}^{N-1}$ on $S^{N-1}(\sqrt{N})$, which means that $\left(E_i(\boldsymbol{\sigma})\right)_{i=1}^{N-1}$ is a local orthonormal basis of $T_{\boldsymbol{\sigma}} S^{N-1}(\sqrt{N})$ for any $\boldsymbol{\sigma}\in S^{N-1}(\sqrt{N})$. We then define the spherical gradient and Hessian of $H_N(\boldsymbol{\sigma})$ as
$$
\nabla_{\mathrm{sp}} H_N(\boldsymbol{\sigma})=\left(E_i H_N(\boldsymbol{\sigma})\right)_{i=1}^{N-1}, \quad \nabla^2_{\mathrm{sp}} H_N(\boldsymbol{\sigma})=\left(E_i E_j H_N(\boldsymbol{\sigma})\right)_{i, j=1}^{N-1} .
$$
Furthermore, let $\partial_r H_N(\boldsymbol{\sigma})$ denote the radial derivative of $H_N(\boldsymbol{\sigma})$ at $\boldsymbol{\sigma} \in S^{N-1}(\sqrt{N})$, which can be understood as the directional derivative of $H_N(\boldsymbol{\sigma})$ in Euclidean space in the direction $\boldsymbol{\sigma}$. For the relationship between $H_N(\boldsymbol{\sigma})$ and its derivatives, we refer the reader to \cite[Lemma 3.2]{Be22}.
To save space, we sometimes also use the shorthand notation $\xi'=\xi'(1)$ and $\xi''=\xi''(1)$. 
\begin{theorem}[Pure $p$-spin model]
    \label{t1}
    Assume $\xi(x)=x^p$ for $p\geq 3$ fixed.
    Let $\boldsymbol{\sigma}^{*}$ denote the global minimum of $H_{N}(\boldsymbol{\sigma})$. Then we have convergence in probability
    \begin{align}
    &\label{emper}\lim_{N\to \8} d(L_{N-1}(\nabla^2_{\mathrm{sp}} H_N(\boldsymbol{\sigma}^*)), \si_{c_p,r_p}) =0,\\
    &\label{eigen}\lim _{N \rightarrow \infty}  \lambda_{\min }\left(\nabla^{2}_{\mathrm{sp}} H_{N}\left(\boldsymbol{\sigma}^{*}\right)\right)=c_p-r_p,\\
    &\label{eigenl}\lim _{N \rightarrow \infty}  \lambda_{\max }\left(\nabla^{2}_{\mathrm{sp}} H_{N}\left(\boldsymbol{\sigma}^{*}\right)\right)=c_p+r_p,
    \end{align}
    where $c_p=pE_0(p)$ with $E_0(p)$ given by (\ref{GS}) and $r_p=2\sqrt{p(p-1)}$.
\end{theorem}
One can easily verify that the pure $p$-spin model belongs to the $pure$-$like$ mixture class for any $p\geq 3$. However, the specific trivial relation (\ref{p}) exists only in the pure $p$-spin model. Therefore, for the sake of clear presentation, we treat the pure $p$-spin case separately from the $pure$-$like$ but not pure $p$-spin case. For the latter case, we require the following condition as stated in \cite{BSZ20}.
\begin{assumption} \label{assum}
 Assume that $\xi(x)$ is mixed and $pure$-$like$, that $\frac{d^2}{d r^2} \Psi_\xi^0(0)<0$, and that the maximum of $\Psi_\xi^0(r)$ on the interval $[-1,1]$ is obtained uniquely at $r=0$.    
 \end{assumption}
 The explicit expression of the function $\Psi_\xi^0(r)$ is defined in \cite[(1.11)]{BSZ20}. Since it is not the key ingredient in this paper, we do not provide it here. This assumption ensures that zero overlap is the unique maximizer of the pair complexity function at levels $-E_0$ and $-y_0$ as defined in (\ref{Ey}). Additionally, it is used to justify the application of the second moment methods, which can be employed to identify the ground state energy. \cite[Theorem 3]{BSZ20} deduced that under this assumption, the model lies in $1$-RSB phase at zero-temperature, but this assumption does not cover the entire zero-temperature $1$-RSB phase \cite[Lemma 3.5]{huang2023constructive}.
 \begin{theorem}[$Pure$-$like$ mixture model]
    \label{t2}
    Under Assumption \ref{assum},
    let $\boldsymbol{\sigma}^{*}$ denote the global minimum of $H_{N}(\boldsymbol{\sigma})$. Then we have convergence in probability
    \begin{align}
    &\label{emper1}\lim_{N\to \8} d(L_{N-1}(\nabla^2_{\mathrm{sp}} H_N(\boldsymbol{\sigma}^*)), \si_{y_0,2\sqrt{\xi''}}) =0,\\
    &\label{eigen1}\lim _{N \rightarrow \infty}  \lambda_{\min }\left(\nabla^{2}_{\mathrm{sp}} H_{N}\left(\boldsymbol{\sigma}^{*}\right)\right)=y_0-2\sqrt{\xi''},\\
    &\label{eigen1l}\lim _{N \rightarrow \infty}  \lambda_{\max }\left(\nabla^{2}_{\mathrm{sp}} H_{N}\left(\boldsymbol{\sigma}^{*}\right)\right)=y_0+2\sqrt{\xi''}
    \end{align}
    where $y_0>2\sqrt{\xi''}$ is defined in (\ref{Ey}).
\end{theorem}
\begin{remark}
Recently, one month after we first posted this article on arXiv, Sellke \cite{sellke2024marginal} deduced the limiting Hessian spectrum for even spins (i.e., $\gamma_p=0$ for all $p$ odd) via a remarkable construction of the unique minimizer of the Parisi functional of vector $p$-spin model \cite{Zhou2022}. In particular, we emphasize that all of our results, Theorems \ref{t1}, \ref{t2} and \ref{t3} can include the odd terms. Furthermore, the techniques used in our work differ from those in \cite{sellke2024marginal} and should be of independent interest.
 \end{remark}
Recall the definition of $1$-RSB model in Definition \ref{defi} and the inequality $y_*\geq 2\sqrt{\xi''}$ given in (\ref{x0y0}). For this model, we have the following result.
 \begin{theorem}[$1$-RSB model]
    \label{t3}
    Assume $\xi(x)$ is $1$-RSB and $y_*>2\sqrt{\xi''}$.
    Let $\boldsymbol{\sigma}^{*}$ denote the global minimum of $H_{N}(\boldsymbol{\sigma})$. Then we have convergence in probability
    \begin{align}
    &\label{emper3}\lim_{N\to \8} d(L_{N-1}(\nabla^2_{\mathrm{sp}} H_N(\boldsymbol{\sigma}^*)), \si_{y_*,2\sqrt{\xi''}}) =0,\\
    &\label{eigen3}\lim _{N \rightarrow \infty}  \lambda_{\min }\left(\nabla^{2}_{\mathrm{sp}} H_{N}\left(\boldsymbol{\sigma}^{*}\right)\right)=y_*-2\sqrt{\xi''},\\
    &\label{eigenl3}\lim _{N \rightarrow \infty}  \lambda_{\max }\left(\nabla^{2}_{\mathrm{sp}} H_{N}\left(\boldsymbol{\sigma}^{*}\right)\right)=y_*+2\sqrt{\xi''}.
    \end{align}
\end{theorem}
\begin{remark}
Here we slightly strengthen the inequality to $y_*>2\sqrt{\xi''}$, as this is required to apply \cite[Lemma 6.4]{Be22} in order to obtain an upper bound for the second moment of the absolute value of the determinant of a shifted GOE matrix. The lemma requires that the shift lies slightly outside the support of the standard semicircle law $\sigma_{0,2}$. This condition $y_*>2\sqrt{\xi''}$ is equivalent to $y>\sqrt{\xi''}$, as defined in (\ref{y}).  
\end{remark}
We organize the remainder of this paper as follows. We prove the convergence of the empirical spectral measures, smallest eigenvalues and largest eigenvalues of the Hessian at global minimum for the spherical pure $p$-spin model in Section \ref{sec2}. Section \ref{sec3} and Section \ref{sec4} are dedicated to the analysis of spherical $pure$-$like$ mixture model and $1$-RSB model, respectively.

\section{Pure $p$-spin model} \label{sec2}
It is more convenient to work with this type of Gaussian process on the $(N-1)$-dimensional unit sphere $S^{N-1}$. For any $\boldsymbol{\sigma} \in S^{N-1}$, we define
\begin{align}\label{hN}
 h_{N}(\boldsymbol{\sigma})= \frac{1}{\sqrt{N}}H_{N}(\sqrt{N} \boldsymbol{\sigma}) .   \end{align}
Then, $ h_{N}(\boldsymbol{\sigma})$ is again a centered Gaussian process with 
$$
\mathbb{E}\left[h_N(\boldsymbol{\sigma}) h_N\left(\boldsymbol{\sigma}^{\prime}\right)\right]= \xi\left(\boldsymbol{\sigma} \cdot \boldsymbol{\sigma}^{\prime}\right).
$$
For any open sets $B\subset \mathbb{R}$ and $D\subset \mathbb{R}$, we define the number of critical points when the value of Hamiltonian is restricted to $\sqrt{N}B$ and the value of radial derivative is restricted to $\sqrt{N}D$
\begin{align}\label{crtbd}
\operatorname{Crt}_N(B,D)=\#\left\{\boldsymbol{\sigma}\in S^{N-1}:\nabla_{\mathrm{sp}} h_N(\boldsymbol{\sigma})=0,\frac{1}{\sqrt{N}}h_N(\boldsymbol{\sigma}) \in B,\frac{1}{\sqrt{N}}\partial_r h_N(\boldsymbol{\sigma})\in D\right\}.
\end{align} 
 By the Kac-Rice formula \cite[theorem 12.1.1]{AT07}, we have
\begin{align} \label{Expect}
\begin{split}
\mathbb{E} \operatorname{Crt}_{N}(B, D)= \omega_N \varphi_{\nabla_{\mathrm{sp}} h_N(\boldsymbol{\sigma})}(0) \mathbb{E}\left\{ \left\lvert \operatorname { d e t } ( \nabla _ { \mathrm { sp } } ^ { 2 } h_{ N } ( \boldsymbol { \sigma } ) ) \right\rvert 1_{\left\{\frac{1}{\sqrt{N}}h_N(\boldsymbol{\sigma}) \in B,\frac{1}{\sqrt{N}}\partial_r h_N(\boldsymbol{\sigma})\in D\right\}} \middle\vert \nabla_{\mathrm{sp}} h_N(\boldsymbol{\sigma})=0\right\},
\end{split}
\end{align}
where $\varphi_{\nabla_{\mathrm{sp}} h_N(\boldsymbol{\sigma})}(0)$ is the density of $\nabla_{\mathrm{sp}} h_N(\boldsymbol{\sigma})$ at $0$,  $\omega_N=2 \pi^{N / 2} / \Gamma(N / 2)$ is the surface area of the $(N-1)$-dimensional unit sphere. 
For the pure $p$-spin model, since the covariance matrix of $h_N(\boldsymbol{\sigma})$ and $\partial_rh_N(\boldsymbol{\sigma})$ is degenerate \cite[Lemma 3.2]{Be22}, and thus we trivially have \begin{align}\label{p}
\partial_rh_N(\boldsymbol{\sigma})=p h_N(\boldsymbol{\sigma}), \quad a.s.  
\end{align}
Moreover, note that the three sets of random variables
$$
h_N(\boldsymbol{\sigma}), \quad \nabla_{\mathrm{sp}} h_N(\boldsymbol{\sigma}), \quad \nabla_{\mathrm{sp}}^2 h_N(\boldsymbol{\sigma})+\partial_rh_N(\boldsymbol{\sigma}) \mathbf{I}_{N-1}
$$
are mutually independent, where $\nabla_{\mathrm{sp}} h_N(\sigma) \sim N\left(0, p \mathbf{I}_{N-1}\right)$, and 
\begin{align}\label{hessian1}
\sqrt{\frac{1}{(N-1)p(p-1)}}\left(\nabla_{\mathrm{sp}}^2 h_N(\boldsymbol{\sigma})+ \partial_rh_N(\boldsymbol{\sigma})\mathbf{I}_{N-1}\right)
\end{align}
is an $(N-1)\times (N-1)$ normalized GOE matrix.
By relations (\ref{p}) and (\ref{hessian1}),
conditioning on $\partial_r h_N(\boldsymbol{\sigma}) = \sqrt{N} y$, we deduce from (\ref{Expect}) that
\begin{align}\label{crip}
\begin{split}
\mathbb{E} \operatorname{Crt}_{N}(B, D)=C_{p, N}
 \int_{D}   \exp \left\{-\frac{Ny^2}{2p^2}\right\} \mathbb{E}\left\{\left|\operatorname{det}\left(\GOE_{N-1}-\frac{\sqrt{N}y}{\sqrt{(N-1)p(p-1)}}\mathbf{I}_{N-1}\right)\right|1_B\left(\frac{y }{p}\right)\right\} \dd y,
 \end{split}
 \end{align}
 where the constant
$$C_{p,N}=\frac{\sqrt{N}(N-1)^{(N-1)/2}(p-1)^{(N-1)/2}}{2^{(N-2)/2}p\Gamma(\frac{N}{2})}.$$
 The complexity function of (\ref{crip}) is obtained in \cite[Lemma 6.4]{Be22} (since the value of external field $h=0$, one can easily check that the supremum over $\gamma$ is at $\gamma=0$), which says that for any open sets $B$ and $D$
\begin{align}\label{logrith}
\begin{split}
\lim _{N \rightarrow \infty} \frac{1}{N} \log \mathbb{E} \Crt_{N}(B, D)
&=\sup _{y \in D\cap pB} \left[ \frac{1}{2}\log (p-1)-\frac{(p-2)y^2}{4p^2(p-1)} + \Phi \left( \frac{y}{\sqrt{2p(p-1)}} \right) \right] \\
&=: \sup _{y \in D\cap pB} R(y),
\end{split}
\end{align}
where $\Phi(x)=\left[ -\frac{|x| \sqrt{x^2-2}}{2}+\ln \left(\frac{|x|+\sqrt{x^2-2}}{\sqrt{2}}\right) \right] 1_{\{|x| \geq \sqrt{2}\}} .$ In fact, due to the relation (\ref{p}), one can check that for any $y\in\mathbb{R}$,
\begin{align}\label{Fp}
 R(py)=\Theta_{p}(y),  
\end{align}
where $\Theta_{p}(y)$ is the complexity function of spherical pure $p$-spin model given in \cite[(2.15)]{ABA13}.
Thus, $R(y)$ is non-decreasing and continuous on $\mathbb{R}$, similar to $\Theta_{p}(y)$. On the other hand, the limiting ground state energy for pure $p$-spin model was computed by Chen and Sen \cite[Proposition 3]{chen}. Together with equality (\ref{gse}), if $\xi(x)=x^p$ for $p\geq 3$, then 
\begin{align}\label{GS}
GS=-E_0(p)=\frac{-\sqrt{p}}{\sqrt{z+1}}\left(1+\frac{z}{p}\right),
\end{align}
where $z>0$ is the unique solution to
\begin{align}\label{pz}
\frac{1}{p}=\frac{1+z}{z^2} \log (1+z)-\frac{1}{z}.    
\end{align}
It follows that
\begin{align}\label{GSlar}
  pE_0(p)>2\sqrt{p(p-1)}.  
\end{align}
Set $f(z)=\frac{\sqrt{p}}{\sqrt{z+1}}(p+z)$ for $z>0$. Then $f'(z)=\frac{\sqrt{p}}{\sqrt{z+1}}-\frac{\sqrt{p}(p+z)}{(z+1)^{3/2}}$ and solving the equation $f'(z)=0$ implies that $f(z)$ attains its unique minimizer at $z=p-2$ with $f(p-2)=2\sqrt{p(p-1)}$. However, $z=p-2$ is clearly not a solution to the equation (\ref{pz}), and thus (\ref{GSlar}) follows since $f(z)=pE_0(p)$ by (\ref{GS}).

For any $\varepsilon>0$ and $\delta>0$,
we define 
\begin{align*}
&\mathcal{A}_{N,p}^{\varepsilon,\delta}=\#\left\{\boldsymbol{\sigma}\in S^{N-1} \colon \nabla_{\mathrm{sp}} h_N(\boldsymbol{\sigma})=0, \; \frac{1}{\sqrt{N}}h_N(\boldsymbol{\sigma}) \in (-E_0(p)-\varepsilon,-E_0(p)+\varepsilon),\right.  \nonumber\\
&\quad\quad\quad\quad\quad\quad \left. \phantom{XXX\frac{1}{\sqrt{N}}} L_{N-1}\left(\nabla^2_{\mathrm{sp}}  h_N(\boldsymbol{\sigma})\right)\notin B\left(\si_{c_p,r_p}, \delta\right)\right\},\\ &
\mathcal{B}_{N,p}^{\varepsilon,\delta}=\#\left\{\boldsymbol{\sigma}\in S^{N-1} \colon \nabla_{\mathrm{sp}} h_N(\boldsymbol{\sigma})=0, \; \frac{1}{\sqrt{N}}h_N(\boldsymbol{\sigma}) \in (-E_0(p)-\varepsilon,-E_0(p)+\varepsilon),\right.  \nonumber\\
&\quad\quad\quad\quad\quad\quad \left. \phantom{XXX\frac{1}{\sqrt{N}}} \left|\lambda_{\min}\left(\nabla^2_{\mathrm{sp}}  h_N(\boldsymbol{\sigma})\right)-c_p+r_p\right|>\delta\right\}\nonumber,\\ &
\mathcal{C}_{N,p}^{\varepsilon,\delta}=\#\left\{\boldsymbol{\sigma}\in S^{N-1} \colon \nabla_{\mathrm{sp}} h_N(\boldsymbol{\sigma})=0, \; \frac{1}{\sqrt{N}}h_N(\boldsymbol{\sigma}) \in (-E_0(p)-\varepsilon,-E_0(p)+\varepsilon),\right.  \nonumber\\
&\quad\quad\quad\quad\quad\quad \left. \phantom{XXX\frac{1}{\sqrt{N}}} \left|\lambda_{\max}\left(\nabla^2_{\mathrm{sp}}  h_N(\boldsymbol{\sigma})\right)-c_p-r_p\right|>\delta\right\}\nonumber.
\end{align*}
\begin{lemma}\label{zero}
 Assume $\xi(x)=x^p$ for some $p\geq 3$ fixed. Then
for any $\delta>0$ fixed, there exists some $\varepsilon>0$ sufficiently small such that  
\begin{align}\label{-infty}
\lim _{N \rightarrow \infty} \frac{1}{N} \log \mathbb{E} \mathcal{A}_{N,p}^{\varepsilon,\delta}=-\infty.
\end{align}
\end{lemma}

\begin{proof}
Using the Kac-Rice formula \cite[theorem 12.1.1]{AT07}, similar to (\ref{crip}), we obtain
\begin{align}\label{criANp}
&\mathbb{E} \mathcal{A}_{N,p}^{\varepsilon,\delta}=\omega_N \varphi_{\nabla_{\mathrm{sp}} h_N(\boldsymbol{\sigma})}(0) \mathbb{E}\left\{ \left\lvert \operatorname { d e t } ( \nabla _ { \mathrm { sp } } ^ { 2 } h_{ N } ( \boldsymbol { \sigma } ) ) \right\rvert \vphantom{1_{\frac{h_N(\boldsymbol{\sigma}) }{\sqrt{N}}}} \right. \nonumber \\
&\quad \left. \times 1_{\left\{\frac{h_N(\boldsymbol{\sigma}) }{\sqrt{N}}\in (-E_0(p)-\varepsilon,-E_0(p)+\varepsilon),\frac{\partial_rh_N(\boldsymbol{\sigma})}{\sqrt{N}} \in  \mathbb{R},L_{N-1}\left(\nabla^2_{\mathrm{sp}}  h_N(\boldsymbol{\sigma})\right)\notin B\left(\si_{c_p,r_p}, \delta\right)\right\}} \middle\vert \nabla_{\mathrm{sp}} h_N(\boldsymbol{\sigma})=0\right\}\nonumber\\&=C_{p,N}
 \int_{-c_p-p\varepsilon}^{-c_p+p\varepsilon} \exp \left\{-\frac{Ny^2}{2p^2}\right\} \mathbb{E}\left\{\left|\operatorname{det}\left(\GOE_{N-1}-\frac{\sqrt{N}y}{\sqrt{(N-1) p(p-1)}}\mathbf{I}_{N-1}\right)\right|1_{A_N^p}\right\} \dd y,    \end{align}
where the event
\begin{align*}
    &A_N^p = \left\{ L_{N-1}\left(\GOE_{N-1}-\frac{\sqrt{N}y}{\sqrt{(N-1) p(p-1)}}\mathbf{I}_{N-1}\right)\right.  \nonumber\\
&\quad\quad\quad\quad\quad\quad \left. \phantom{XXX\frac{1}{\sqrt{N}}}\notin B\left(\si_{\frac{\sqrt{N}c_p}{\sqrt{(N-1)p(p-1)}},\frac{2\sqrt{N}}{\sqrt{N-1}}}, \frac{\sqrt{N}\delta}{\sqrt{(N-1)p(p-1)}}\right) \right\}.
\end{align*}
Since $c_p=pE_0(p)>2\sqrt{p(p-1)}$ by (\ref{GSlar}), we may choose $\varepsilon>0$ to be sufficiently small such that for all $N$ large enough and any $y\in(-c_p-p\varepsilon,-c_p+p\varepsilon)$, we have
\begin{align}\label{y}
\frac{\sqrt{N}y}{\sqrt{(N-1) p(p-1)}}<-2.
\end{align}
Note that for $|y+c_p|\leq p\varepsilon$ and a function $h$ with $\|h\|_{L} \leq 1$, for any $N$ large enough
\begin{align}\label{Lip}
&\left\lvert \frac{1}{N-1} \sum_{i=1}^{N-1}\left[ h\left(\lambda_{i}(\GOE_{N-1})-\frac{\sqrt{N}(y+c_p)}{\sqrt{(N-1) p(p-1)}}\right)- h(\lambda_{i}\left(\GOE_{N-1})\right)\right] \right\rvert \nonumber \\
&\leq \frac{\sqrt{Np}\varepsilon}{\sqrt{(N-1)(p-1)}}\leq 2\varepsilon.
\end{align}
Writing $c_p'=\frac{c_p}{\sqrt{p(p-1)}}$, since $\frac{\sqrt{N}c_p}{\sqrt{(N-1)p(p-1)}}\rightarrow c_p'$, $\frac{2\sqrt{N}}{\sqrt{N-1}}\rightarrow 2$ and $ \frac{\sqrt{N}\delta}{\sqrt{(N-1)p(p-1)}}\rightarrow \frac{\delta}{\sqrt{p(p-1)}}$ as $N\rightarrow \infty$, it follows that $$d\left(\si_{\frac{\sqrt{N}c_p}{\sqrt{(N-1)p(p-1)}},\frac{2\sqrt{N}}{\sqrt{N-1}}},\si_{c_p',2}\right)\xrightarrow{N \rightarrow \infty}0.$$ By the triangle inequality, for any $N$ large enough and $\varepsilon<\frac{\delta}{6\sqrt{p(p-1)}}$, if the event $A_N^p$ occurs, we deduce 
\begin{align}\label{triangle}
    \mathrel{\phantom{=}} \frac{2\delta}{3\sqrt{p(p-1)}}&< d\left(L_{N-1}\left(\GOE_{N-1}-\frac{\sqrt{N}y}{\sqrt{(N-1) p(p-1)}}\mathbf{I}_{N-1}\right), \si_{c_p',2}\right)\nonumber \\&= d\left(L_{N-1}\left(\GOE_{N-1}- \frac{\sqrt{N}(y+c_p)}{\sqrt{(N-1) p(p-1)}} \mathbf{I}_{N-1}\right), \si_{c_p'-\frac{\sqrt{N}c_p}{\sqrt{(N-1)p(p-1)}},2}\right)\nonumber\\ &\le d\left(L_{N-1}\left(\GOE_{N-1}- \frac{\sqrt{N}(y+c_p)}{\sqrt{(N-1) p(p-1)}}\mathbf{I}_{N-1}\right), L_{N-1}(\GOE_{N-1})\right)\nonumber \\&\quad+d(L_{N-1}(\GOE_{N-1}), \si_{\rm sc})+d\left(\si_{c_p'-\frac{\sqrt{N}c_p}{\sqrt{(N-1)p(p-1)}},2}, \si_{\rm sc}\right),
\end{align}
which implies $d(L_{N-1}(\GOE_{N-1}), \si_{\rm sc})>\frac{\delta}{3\sqrt{p(p-1)}}=:\delta'$ by (\ref{Lip}). Together with the Cauchy--Schwarz inequality, we find for $N$ large enough
\begin{align}\label{criANp1}
\mathbb{E} \mathcal{A}_{N,p}^{\varepsilon,\delta}&\leq C_{p,N}
 \int_{-c_p-p\varepsilon}^{-c_p+p\varepsilon}\exp \left\{-\frac{Ny^2}{2p^2}\right\} \nonumber\\
 &\quad \times\mathbb{E}\left\{\left|\operatorname{det}\left(\GOE_{N-1}-\frac{\sqrt{N}y}{\sqrt{(N-1) p(p-1)}}\mathbf{I}_{N-1}\right)\right|1_{\left\{L_{N-1}\left(\GOE_{N-1}\right)\notin B\left(\si_{\rm{sc}}, \delta'\right)\right\}}\right\} \dd y \nonumber\\
 &\leq C_{p,N}
 \int_{-c_p-p\varepsilon}^{-c_p+p\varepsilon}\exp\left\{-\frac{Ny^2}{2p^2}\right\}\mathbb{P}\left(L_{N-1}\left(\GOE_{N-1}\right)\notin B\left(\si_{\rm{sc}}, \delta'\right)\right)^{\frac{1}{2}} \nonumber\\
 &\quad\quad\quad\times\left[\mathbb{E}\left\{\left|\operatorname{det}\left(\GOE_{N-1}-\frac{\sqrt{N}y}{\sqrt{(N-1) p(p-1)}}\mathbf{I}_{N-1}\right)\right|^2\right\}\right]^{\frac{1}{2}} \dd y.   \end{align}
 By symmetry, the property (\ref{y}) allows us to apply the upper bound for the second moment of the determinant of a shifted GOE matrix \cite[Lemma 6.4]{Be22} (note the different scaling used in the definition of the standard GOE matrix). This implies that  for large enough $N$ and small enough $\varepsilon$, it holds for any $\gamma>0$ and all $y\in(-c_p-p\varepsilon,-c_p+p\varepsilon)$ that
 \begin{align}\label{second}
 &\mathbb{E}\left\{\left|\operatorname{det}\left(\GOE_{N-1}-\frac{\sqrt{N}y}{\sqrt{(N-1) p(p-1)}}\mathbf{I}_{N-1}\right)\right|^2\right\}\nonumber\\&\leq e^{\gamma N}\mathbb{E}\left\{\left|\operatorname{det}\left(\GOE_{N-1}-\frac{\sqrt{N}y}{\sqrt{(N-1) p(p-1)}}\mathbf{I}_{N-1}\right)\right|\right\}^2.    \end{align}
 Since the empirical spectral measures of the GOE matrices satisfy an LDP with speed $N^2$ \cite{benarous1997large}, there exists a constant $C>0$ such that for $N$ large enough
\begin{align}\label{devia}
\mathbb{P}\left(L_{N-1}\left(\GOE_{N-1}\right) \notin B\left(\sigma_{\mathrm{sc}}, \delta'\right)\right) \leq e^{-C N^2},
\end{align}
where $B(\sigma_{\mathrm{sc}}, \delta')$ denotes the open ball in the space $\mathcal{P}(\mathbb{R})$ with center $\sigma_{\mathrm{sc}}$ and radius $\delta'$ with respect to the metric $d$ given by \pref{eq:measd}. 
Plugging the above two estimates (\ref{second}) and (\ref{devia}) into (\ref{criANp1}) results in 
\begin{align*}
\mathbb{E} \mathcal{A}_{N,p}^{\varepsilon,\delta}&\leq C_{p,N} 
 \int_{-c_p-p\varepsilon}^{-c_p+p\varepsilon}\exp\left\{-C'N^2\right\}\exp\left\{-\frac{Ny^2}{2p^2}\right\}\nonumber\\&\quad\quad\quad\times\mathbb{E}\left\{\left|\operatorname{det}\left(\GOE_{N-1}-\frac{\sqrt{N}y}{\sqrt{(N-1) p(p-1)}}\mathbf{I}_{N-1}\right)\right|\right\} \dd y,  \end{align*}
 for some constant $C'>0$ and all $N$ large enough. From (\ref{logrith}), we obtain 
 \begin{align*}
  \mathbb{E} \mathcal{A}_{N,p}^{\varepsilon,\delta}\leq \exp\{-C'N^2  \} \exp\left\{\sup_{y \in (-c_p-p\varepsilon,-c_p+p\varepsilon)}NR(y)+o(N)\right\},  
  \end{align*}
  which directly gives (\ref{-infty}) by the continuity of $R(y)$ and the fact $R(-c_p)=0$.
  \end{proof}
  \begin{lemma}\label{zero1}
Assume $\xi(x)=x^p$ for some $p\geq 3$ fixed. Then
for any $\delta>0$ fixed, there exist some $\varepsilon>0$ sufficiently small and some $\eta>0$ such that
\begin{align}\label{-eta}
&\limsup _{N \rightarrow \infty} \frac{1}{N} \log \mathbb{E} \mathcal{B}_{N,p}^{\varepsilon,\delta}\leq-\eta \\ & \label{-eta1}\limsup _{N \rightarrow \infty} \frac{1}{N} \log \mathbb{E} \mathcal{C}_{N,p}^{\varepsilon,\delta}\leq-\eta.
\end{align}
\end{lemma}
\begin{proof}
The proof of this lemma follows a similar strategy as in Lemma \ref{zero}. We need to address equality (\ref{criANp}) with the event $A_N^p$ replaced by the event
$$
B_N^p = \left\{ \left| \lambda_{\min}\left(\GOE_{N-1}\right)-\frac{\sqrt{N}(y+c_p)}{\sqrt{(N-1) p(p-1)}}+\frac{2\sqrt{N}}{\sqrt{N-1}}\right|> \frac{\sqrt{N}\delta}{\sqrt{(N-1)p(p-1)}} \right\}.
$$
For any fixed $\delta>0$, there exists a $\varepsilon>0$ sufficiently small such that for any $N$ sufficiently large, the event $B_N^p$ implies the event
$$\left| \lambda_{\min}\left(\GOE_{N-1}\right)+2\right|> \frac{\delta}{2\sqrt{p(p-1)}},$$
and thus we have
\begin{align}\label{criANp2}
\mathbb{E} \mathcal{B}_{N,p}^{\varepsilon,\delta}&\leq C_{p,N}
 \int_{-c_p-p\varepsilon}^{-c_p+p\varepsilon}\exp \left\{-\frac{Ny^2}{2p^2}\right\} \nonumber\\
 &\times\mathbb{E}\left\{\left|\operatorname{det}\left(\GOE_{N-1}-\frac{\sqrt{N}y}{\sqrt{(N-1) p(p-1)}}\mathbf{I}_{N-1}\right)\right|1_{\left\{\left| \lambda_{\min}\left(\GOE_{N-1}\right)+2\right|> \frac{\delta}{2\sqrt{p(p-1)}}\right\}}\right\} \dd y. 
 \end{align}
 The LDP for the smallest eigenvalues of the GOE matrices \cite{BDG01} implies that for any $\delta>0$, there exists a $\zeta>0$ such that for all $N$ large enough
\begin{align}
\label{el}
\mathbb{P}\left(\left|\lambda_{\min}\left(\GOE_N\right)+2\right|>\frac{\delta}{2\sqrt{p(p-1)}}\right) \leq e^{-\zeta N}.
\end{align}
Note that the upper bound (\ref{second}) holds for any $\gamma>0$. Here we take $\gamma<\frac{\zeta}{2}$. Hence, 
 combining the Cauchy–Schwarz inequality, (\ref{second}) and (\ref{el}), we have for all $\varepsilon>0$ sufficiently small and $N$ large enough
 \begin{align*}
\mathbb{E} \mathcal{B}_{N,p}^{\varepsilon,\delta}&\leq C_{p,N}
 \int_{-c_p-p\varepsilon}^{-c_p+p\varepsilon}\exp\left\{-\frac{\zeta N}{2}\right\}\exp\left\{-\frac{Ny^2}{2p^2}\right\}\nonumber\\&\quad\quad\quad\times\mathbb{E}\left\{\left|\operatorname{det}\left(\GOE_{N-1}-\frac{\sqrt{N}y}{\sqrt{(N-1) p(p-1)}}\mathbf{I}_{N-1}\right)\right|\right\} \dd y,  \end{align*}
 which yields the desired result (\ref{-eta}) for any $\eta<\frac{\zeta}{2}$ by the logarithmic asymptotics formula (\ref{logrith}), the fact $R(-c_p)=0$ and the continuity of $R(y)$. The proof for (\ref{-eta1}) is similar, since the LDP for the largest eigenvalues of the GOE matrices \cite{BDG01} states that, for any $\delta>0$, there exists $\zeta>0$ such that for $N$ sufficiently large,
\begin{align*}
\mathbb{P}\left(\left|\lambda_{\max}\left(\GOE_N\right)-2\right|>\frac{\delta}{2\sqrt{p(p-1)}}\right) \leq e^{-\zeta N}.
\end{align*}  
\end{proof}
Now we are ready to prove Theorem \ref{t1}. 
\begin{proof}[Proof of Theorem \ref{t1}] 
Note that $GS=-E_0(p)$ from (\ref{GS}) for the pure $p$-spin model.
 To prove (\ref{emper}), it suffices to show that for any $\delta>0$,
 \begin{align*}
  \mathbb{P}(L_{N-1}(\nabla^2_{\mathrm{sp}} H_N(\boldsymbol{\sigma}^*))\notin B(\si_{c_p,r_p},\delta)) \xrightarrow{N \rightarrow \infty} 0. 
 \end{align*}
 It is clear that for any $\delta>0$ and $\varepsilon>0$
 \begin{align}\label{innotin}
\begin{split}
  &\mathrel{\phantom{=}} \mathbb{P} \left( L_{N-1}(\nabla^2_{\mathrm{sp}} H_N(\boldsymbol{\sigma}^*))\notin B(\si_{c_p,r_p},\delta) \right) \\
  &= \mathbb{P} \left( L_{N-1}(\nabla^2_{\mathrm{sp}} H_N(\boldsymbol{\sigma}^*))\notin B(\si_{c_p,r_p},\delta),\frac{1}{N}H_N(\boldsymbol{\sigma}^*)\in(-E_0(p)-\varepsilon,-E_0(p)+\varepsilon) \right) \\
  &\quad+ \mathbb{P} \left( L_{N-1}(\nabla^2_{\mathrm{sp}} H_N(\boldsymbol{\sigma}^*))\notin B(\si_{c_p,r_p},\delta),\frac{1}{N}H_N(\boldsymbol{\sigma}^*)\notin(-E_0(p)-\varepsilon,-E_0(p)+\varepsilon) \right).
  \end{split}
  \end{align}
 Since $\frac{1}{N}H_N(\boldsymbol{\sigma}^*)$ converges to $-E_0(p)$ almost surely as $N$ tends to infinity, the second term on the right-hand side of (\ref{innotin}) goes to zero as $N$ tends to infinity for any $\varepsilon>0$. If the event in the first term on the right-hand side of (\ref{innotin}) occurs, there exists at least one critical points $\boldsymbol{\sigma}\in S^{N-1}(\sqrt{N})$ with the restriction \begin{align}\label{HHessian}
     \frac{1}{N}H_N(\boldsymbol{\sigma})\in(-E_0(p)-\varepsilon,-E_0(p)+\varepsilon) ,\quad L_{N-1}(\nabla^2_{\mathrm{sp}} H_N(\boldsymbol{\sigma}))\notin B(\si_{c_p,r_p},\delta).
     \end{align}
Due to the different scaling on the domain (\ref{hN}), one has for any $\boldsymbol{\sigma} \in S^{N-1}$,
\begin{align*}
\nabla^{2}_{\mathrm{sp}} H_{N}\left(\sqrt{N}\boldsymbol{\sigma}\right)= \nabla^{2}_{\mathrm{sp}} h_{N}\left(\boldsymbol{\sigma}\right).    
\end{align*} 
Thus, the above event (\ref{HHessian}) is equivalent to the event that there exists at least one critical points $\boldsymbol{\sigma}\in S^{N-1}$ with $$\frac{1}{\sqrt{N}}h_N(\boldsymbol{\sigma})\in(-E_0(p)-\varepsilon,-E_0(p)+\varepsilon) ,\quad L_{N-1}(\nabla^2_{\mathrm{sp}} h_N(\boldsymbol{\sigma}))\notin B(\si_{c_p,r_p},\delta).$$
 Therefore, for any $\delta>0$ and $\varepsilon>0$, we have
 \begin{align*}
&\mathbb{P} \left( L_{N-1}(\nabla^2_{\mathrm{sp}} H_N(\boldsymbol{\sigma}^*))\notin B(\si_{c_p,r_p},\delta),\frac{1}{N}H_N(\boldsymbol{\sigma}^*)\in(-E_0(p)-\varepsilon,-E_0(p)+\varepsilon) \right)\nonumber\\& \leq \mathbb{P}(\mathcal{A}_{N,p}^{\varepsilon,\delta}\geq 1) .  
\end{align*}
 However, Lemma \ref{zero} together with Markov's inequality implies that for any $\delta>0$, there exists some $\varepsilon>0$ small enough such that
 \begin{align*}
\lim_{N\rightarrow\8}\mathbb{P}(\mathcal{A}_{N,p}^{\varepsilon,\delta}\geq 1)=0, 
 \end{align*} 
 which completes the proof of (\ref{emper}). The argument above can also be applied to obtain (\ref{eigen}) and (\ref{eigenl}), since by Markov’s inequality, Lemma \ref{zero1} directly gives $\lim_{N\rightarrow\8}\mathbb{P}(\mathcal{B}_{N,p}^{\varepsilon,\delta}\geq 1)=0$ and $\lim_{N\rightarrow\8}\mathbb{P}(\mathcal{C}_{N,p}^{\varepsilon,\delta}\geq 1)=0$. 
 \end{proof}

\section{Pure-like Mixture model} \label{sec3}

 Based on the covariance structure calculated in \cite[Lemma 3.2]{Be22}, the three sets of random variables
$$
\left(h_N(\boldsymbol{\sigma}), \partial_r h_N(\boldsymbol{\sigma})\right), \quad
\nabla_{\mathrm{sp}} h_N(\boldsymbol{\sigma}), \quad
\nabla_{\mathrm{sp}}^2 h_N(\boldsymbol{\sigma})+\partial_r h_N(\boldsymbol{\sigma}) \mathbf{I}_{N-1}
$$
are mutually independent, where $\nabla_{\mathrm{sp}} h_N(\sigma) \sim N\left(0, \xi^{\prime}\left(1\right) \mathbf{I}_{N-1}\right)$, and
\begin{align}\label{hessian}
\sqrt{\frac{1}{(N-1) \xi^{\prime \prime}\left(1\right)}}\left(\nabla_{\mathrm{sp}}^2 h_N(\boldsymbol{\sigma})+ \partial_r h_N(\boldsymbol{\sigma}) \mathbf{I}_{N-1}\right)
\end{align}
is an $(N-1)\times (N-1)$ GOE matrix and
$$
\left( h_N(\boldsymbol{\sigma}), \partial_r h_N(\boldsymbol{\sigma})\right) \sim N\left(0, \Sigma\right),
$$
with
\begin{equation}\label{sigma}
\Sigma:=\left(\begin{array}{cc}
\xi\left(1\right) &  \xi^{\prime}\left(1\right) \\ \xi^{\prime}\left(1\right) &  \xi^{\prime \prime}\left(1\right)+\xi^{\prime}\left(1\right)
\end{array}\right).
\end{equation}
From (\ref{Expect}), conditioning on $\frac{h_N(\boldsymbol{\sigma})}{\sqrt{N}}=x$ and $\frac{\partial_r h_N(\boldsymbol{\sigma})}{\sqrt{N}}=y$, we immediately arrive at
\begin{align}\label{criBD}
\mathbb{E} \operatorname{Crt}_{N}(B, D)&=C_{\xi,N}
 \int_{B \times D}   \exp \left\{-\frac{N}{2}(x, y) \Sigma^{-1}(x, y)^{\top}\right\}\nonumber\\&\quad\quad\quad\times \mathbb{E}\left\{\left|\operatorname{det}\left(\GOE_{N-1}-\frac{\sqrt{N}y}{\sqrt{(N-1) \xi^{\prime \prime}\left(1\right)}}\mathbf{I}_{N-1}\right)\right|\right\}\dd x \dd y,    
\end{align}
where $$C_{\xi,N}=\frac{\sqrt{N}(N-1)^{(N-1)/2}\xi''(1)^{(N-1)/2}}{2^{(N-1)/2}\Gamma(\frac{N}{2})\xi'(1)^{(N-1)/2}\sqrt{\pi \det \Sigma}}.$$
Let us define
$$
\Psi_{*}(x)=\int_{\mathbb{R}} \log |x-t|\sigma_{\mathrm{sc}} (\dd t).
$$
Direct calculation yields
\begin{align}\label{psi}
\Psi_{*}(x) = \begin{cases}\frac{x^2}{4}-\frac{1}{2} & \text { if } 0 \leq|x| \leq 2, \\ \frac{x^2}{4}-\frac{1}{2}-\left[\frac{|x|}{4} \sqrt{x^2-4}-\log \left(\sqrt{\frac{x^2}{4}-1}+\frac{|x|}{2}\right)\right] & \text { if }|x|>2.\end{cases}
\end{align}
Based on $\Psi_{*}(x)$, we define
\begin{align}\label{fxy}
   F(x,y) &= \frac{1}{2}+\frac{1}{2} \log \left(\frac{\xi^{\prime \prime}(1)}{\xi^{\prime}(1)}\right)-\frac{1}{2}(x, y) \Sigma^{-1}(x, y)^\top + \Psi_{*}\left(\frac{y}{ \sqrt{\xi^{\prime \prime}(1)}}\right)
\nonumber \\
&= \frac{1}{2}+\frac{1}{2} \log \left(\frac{\xi^{\prime \prime}(1)}{\xi^{\prime}(1)}\right)-\frac{x^2}{2} -\frac{\left(y-\xi^{\prime}(1) x\right)^2}{2\left(\xi^{\prime \prime}(1)+\xi^{\prime}(1)-\xi^{\prime}(1)^2\right)}+\Psi_*\left(\frac{y}{\sqrt{\xi^{\prime \prime}(1)}}\right).
\end{align}
It was established in \cite[Theorem 3.1]{BSZ20} that
for any intervals $B$ and $D$, 
\begin{align}\label{supBD}
\lim _{N \rightarrow \infty} \frac{1}{N} \log \mathbb{E} \Crt_{N}(B, D)=\sup _{x \in B, y \in D} F(x, y) .    
\end{align}
To maintain consistency in these series works, we define
\begin{align}\label{Ey}
-E_0=\min \left\{x: \sup _{y \in \mathbb{R}} F(x, y)=0\right\},\quad -y_0= \operatornamewithlimits{argmax}_{y \in \mathbb{R}} F\left(-E_0, y\right) .
\end{align}
Here we emphasize that for the $pure$-$like$ mixture model, $E_0$ is equal to $E_0(\xi)$ mentioned in the introduction section, as claimed in Footnote $2$ of \cite{BSZ20}.
The authors \cite[Theorem 1.1]{BSZ20} proved that the sequence of ground state energy $GS_N$ converges to $-E_0$ almost surely under Assumption \ref{assum}, that is, 
\begin{align}\label{gsm}
\lim _{N \rightarrow \infty}GS_N =-E_0, \quad a.s.
\end{align}
 We introduce an important threshold of energy level 
\begin{equation*}
E_{\infty}=\frac{\xi^{\prime \prime}(1) +\xi^{\prime}(1)^2-\xi^{\prime}(1) }{\xi^{\prime}(1) \sqrt{\xi^{\prime \prime}(1)}}.
\end{equation*}
The value $E_{\infty}$ is related to $G\left(\xi^{\prime}, \xi^{\prime \prime}\right)$ defined in (\ref{Gxi}) by $G\left(\xi^{\prime}, \xi^{\prime \prime}\right)=\sup _{y \in \mathbb{R}} F(-E_{\infty}, y)$ \cite[(4.1)]{ABA13}.
From \cite[Lemmas 5.1]{BSZ20}, due to $-E_0<-E_{\8}$, we know that if $\xi(x)$ belongs to 
the $pure$-$like$ mixture class, the solution $-y_0$ in 
(\ref{Ey}) is unique and
\begin{align}\label{-y0}
-y_0<-2\sqrt{\xi''(1)}.
\end{align}
Later, we will prove that for the $pure$-$like$ mixture case, $y_0$ determines the center of the limiting empirical spectral measure of the Hessian at global minimum $\boldsymbol{\sigma}^*$.
 
Recall the function $F(x,y)$ defined in (\ref{fxy}). By the definition of $-E_0$ and $-y_0$, since $-y_0<-2\sqrt{\xi''(1)}$ from (\ref{-y0}), we know $F(-E_0,-y_0)=0$ and
\begin{align*}
  \frac{\partial F}{\partial y} (-E_0, -y_0)=  - \frac{-y_0 + \xi' (1) E_0}{\xi'' (1) + \xi' (1) - \xi' (1)^2} + \frac{-y_0 + \sqrt{y_0^2 - 4 \xi'' (1)}}{2 \xi'' (1)}=0. 
\end{align*}
We establish some properties of $F(x,y)$ when $x$ is in a small neighborhood of $-E_0$, which is crucial for the subsequent proof. Morally speaking, the following lemma implies that the $\argmax$ function defined in (\ref{Ey}), is continuous with respect to $x$ when $x$ lies in a small neighborhood of $-E_0$.
\begin{lemma}\label{ls}
If $\xi(x)$ is a pure-like mixture, then there exists some constant $\alpha > 0$ such that for sufficiently small $\varepsilon>0$,
\begin{align}\label{smallzero}
    F (x, y) < 0, \quad \forall x \in (-E_0 - \varepsilon, -E_0 + \varepsilon), \; y \in \mathbb{R} \setminus (- y_0 - \alpha \sqrt{\varepsilon}, - y_0 + \alpha \sqrt{\varepsilon}),
\end{align}
 and
there exists some constant $\beta > 0$ such that for sufficiently small $\varepsilon>0$,
\begin{align}\label{smallvar}
    F (x, y) < \beta \sqrt{\varepsilon}, \quad \forall x \in (-E_0 - \varepsilon, -E_0 + \varepsilon), \; y \in (- y_0 - \alpha \sqrt{\varepsilon}, - y_0 + \alpha \sqrt{\varepsilon}).
\end{align}
\end{lemma}

\begin{proof}
We first show that there exists $K < -2 \sqrt{\xi'' (1)}$ such that $F (x, y) < 0$ for all $x \in (-E_0 - \varepsilon, -E_0 + \varepsilon)$ and any $y < K$. Then we prove it for $y \in [K, 0] \setminus (- y_0 - \alpha \sqrt{\varepsilon}, - y_0 + \alpha \sqrt{\varepsilon})$ and $y > 0$ successively. Without loss of generality, assume that $\varepsilon < E_0$.
According to \eqref{psi}, for $\lvert x \rvert > 2$,
\begin{align*}
    \Psi_{*} (x)
    &= \frac{x^2}{4} - \frac{1}{2} - \left[ \frac{\lvert x \rvert}{4} \sqrt{x^2 - 4} - \log \left( \sqrt{\frac{x^2}{4}-1} + \frac{\lvert x \rvert}{2} \right) \right] \\
    &= - \frac{1}{2} + \frac{x^2 - \lvert x \rvert \sqrt{x^2 - 4}}{4} + \log \left( \frac{\lvert x \rvert + \sqrt{x^2 - 4}}{2} \right) \\
    &= - \frac{1}{2} + \frac{\lvert x \rvert}{\lvert x \rvert + \sqrt{x^2 - 4}} - \log \left( \frac{\lvert x \rvert}{\lvert x \rvert + \sqrt{x^2 - 4}} \right) + \log \left( \frac{\lvert x \rvert}{2} \right).
\end{align*}
Since $t := \lvert x \rvert / (\lvert x \rvert + \sqrt{x^2 - 4}) \in (1 / 2, 1)$, we have $t - \log t < 1 / 2 + \log 2$, which implies that
\begin{equation*}
    \Psi_{*} (x) < - \frac{1}{2} + \frac{1}{2} + \log 2 + \log \left( \frac{\lvert x \rvert}{2} \right) = \log \lvert x \rvert, \quad \text{for } \lvert x \rvert > 2.
\end{equation*}
If $\lvert x + E_0 \rvert < \varepsilon$ and $y < \min \{ -2 \sqrt{\xi'' (1)}, -2 E_0 \xi' (1) \}$, we have
\begin{align*}
    F (x, y) < \frac{1}{2} + \frac{1}{2} \log \left( \frac{\xi'' (1)}{\xi' (1)} \right) - \frac{(y + 2 \xi' (1) E_0)^2}{2 (\xi'' (1) + \xi' (1) - \xi' (1)^2)} + \log \left\lvert \frac{y}{\sqrt{\xi'' (1)}} \right\rvert.
\end{align*}
Therefore, for sufficiently small $\varepsilon > 0$, there exists $K < -2 \sqrt{\xi'' (1)}$ sufficiently small, such that $F (x, y) < 0$ for any $\lvert x + E_0 \rvert < \varepsilon$ and $y < K$.

For $x \in (-E_0 - \varepsilon, -E_0 + \varepsilon)$, $y \in [K, 0]$,
\begin{align} \label{eq:1.5pr1}
    F (x, y) - F (-E_0, y) = - \frac{x^2 - E_0^2}{2} - \frac{(y - \xi' (1) x)^2 - (y + \xi' (1) E_0)^2}{2 (\xi'' (1) + \xi' (1) - \xi' (1)^2)} \leq C \varepsilon,
\end{align}
where
\begin{align*}
    C = \frac{3 E_0}{2} + \frac{-3 K \xi' (1)E_0 + 3 \xi' (1)^2 E_0^2}{2 (\xi'' (1) + \xi' (1) - \xi' (1)^2)}.
\end{align*}
It remains to estimate $F (-E_0, y)$. Straightforward calculations show that
\begin{align*}
    \frac{\partial F}{\partial y} (-E_0, y) = \begin{cases}
    - \frac{y + \xi' (1) E_0}{\xi'' (1) + \xi' (1) - \xi' (1)^2} + \frac{y}{2 \xi'' (1)}, & -2 \sqrt{\xi'' (1)} \leq y \leq 0, \\
    - \frac{y + \xi' (1) E_0}{\xi'' (1) + \xi' (1) - \xi' (1)^2} + \frac{y + \sqrt{y^2 - 4 \xi'' (1)}}{2 \xi'' (1)}, & y < -2 \sqrt{\xi'' (1)}.
    \end{cases}
\end{align*}
which is a decreasing function on $(-\infty, 0]$ with $\partial F / \partial y (-E_0, -y_0) = 0$. For $y < -2 \sqrt{\xi'' (1)}$,
\begin{align*}
    \frac{\partial^2 F}{\partial y^2} (-E_0, y) = - c + \frac{y}{2 \xi'' (1) \sqrt{y^2 - 4 \xi'' (1)}} \leq -c, \quad 
\end{align*}
where $c = \frac{\xi'' (1) - \xi' (1) + \xi' (1)^2}{2 \xi'' (1) (\xi'' (1) + \xi' (1) - \xi' (1)^2)} > 0$.
For any $\varepsilon < (y_0 - 2 \sqrt{\xi'' (1)})^2$, it follows that
\begin{align*}
    \frac{\partial F}{\partial y} (-E_0, -y_0 - \sqrt{\varepsilon}) &= \frac{\partial F}{\partial y} (-E_0, -y_0 - \sqrt{\varepsilon}) - \frac{\partial F}{\partial y} (-E_0, -y_0) \geq c \sqrt{\varepsilon}, \\
    \frac{\partial F}{\partial y} (-E_0, -y_0 + \sqrt{\varepsilon}) &= \frac{\partial F}{\partial y} (-E_0, -y_0 + \sqrt{\varepsilon}) - \frac{\partial F}{\partial y} (-E_0, -y_0) \leq -c \sqrt{\varepsilon}.
\end{align*}
Since $\partial F / \partial y (-E_0, \cdot)$ is decreasing on $(- \infty, 0]$, for $\alpha > 1$, we have
\begin{align*}
    F (-E_0, -y_0 - \alpha \sqrt{\varepsilon})
    &= F (-E_0, -y_0 - \sqrt{\varepsilon}) - \int_{-y_0 - \alpha \sqrt{\varepsilon}}^{-y_0 - \sqrt{\varepsilon}} \frac{\partial F}{\partial y} (-E_0, y) \, \mathrm{d} y
    \leq -c (\alpha - 1) \varepsilon, \\
    F (-E_0, -y_0 + \alpha \sqrt{\varepsilon})
    &= F (-E_0, -y_0 + \sqrt{\varepsilon}) + \int_{-y_0 + \sqrt{\varepsilon}}^{-y_0 + \alpha \sqrt{\varepsilon}} \frac{\partial F}{\partial y} (-E_0, y) \, \mathrm{d} y
    \leq -c (\alpha - 1) \varepsilon.
\end{align*}
Together with the monotonicity of $F (-E_0, \cdot)$, we have
\begin{align} \label{eq:1.5pr2}
    F (-E_0, y) \leq -c (\alpha - 1) \varepsilon, \quad \forall y \in [K, 0] \setminus (- y_0 - \alpha \sqrt{\varepsilon}, - y_0 + \alpha \sqrt{\varepsilon}).
\end{align}
Let $\alpha > 1 + C / c$. Then \eqref{eq:1.5pr1} and \eqref{eq:1.5pr2} imply that $F (x, y) < 0$ for any $x \in (-E_0 - \varepsilon, -E_0 + \varepsilon)$ and $y \in [K, 0] \setminus (- y_0 - \alpha \sqrt{\varepsilon}, - y_0 + \alpha \sqrt{\varepsilon})$.

For $y > 0$, since $\Psi_*$ is an even function,
\begin{align*}
    F (x, y) = F (x, -y) + \frac{(y + \xi' (1) x)^2 - (y - \xi' (1) x)^2}{2 (\xi'' (1) + \xi' (1) - \xi' (1)^2)}.
\end{align*}
For $x < 0$, it holds that $(y + \xi' (1) x)^2 < (y - \xi' (1) x)^2$, which further implies that
\begin{align*}
    F (x, y) < F (x, -y) < 0, \quad \forall x \in (-E_0 - \varepsilon, -E_0 + \varepsilon), \; y \in (0, +\infty) \setminus (y_0 - \alpha \sqrt{\varepsilon}, y_0 + \alpha \sqrt{\varepsilon}).
\end{align*}
If $y \in (y_0 - \alpha \sqrt{\varepsilon}, y_0 + \alpha \sqrt{\varepsilon})$,
\begin{align*}
    (y + \xi' (1) x)^2 - (y - \xi' (1) x)^2
    = 4 \xi' (1) x y
    \leq 4 \xi' (1) (-E_0 + \varepsilon) (y_0 - \alpha \sqrt{\varepsilon})
    < 0.
\end{align*}
According to \eqref{eq:1.5pr1}, $F (x, -y) \leq C \varepsilon$, so $F (x, y) < 0$ still holds with sufficiently small $\varepsilon$, and thus we finish the proof of claim (\ref{smallzero}).

For $x \in (-E_0 - \varepsilon, -E_0 + \varepsilon)$, $y \in (- y_0 - \alpha \sqrt{\varepsilon}, - y_0 + \alpha \sqrt{\varepsilon})$,
\begin{align*} 
    F (x, y) - F (-E_0, -y_0)& = - \frac{x^2 - E_0^2}{2} - \frac{(y - \xi' (1) x)^2 - (-y_0 + \xi' (1) E_0)^2}{2 (\xi'' (1) + \xi' (1) - \xi' (1)^2)}\nonumber\\&\quad +\Psi_*\left(\frac{y}{\sqrt{\xi^{\prime \prime}(1)}}\right)-\Psi_*\left(\frac{-y_0}{\sqrt{\xi^{\prime \prime}(1)}}\right).
\end{align*}
Note that the function $\Psi_*$ is Lipschitz and denote by $L_{\Psi_*}$ the Lipschitz constant. A little algebra yields that the claim (\ref{smallvar}) holds for sufficiently small $\varepsilon>0$ with
\begin{align*}
\beta &> \frac{3\alpha (y_0+\xi'E_0)}{2 (\xi'' (1) + \xi' (1) - \xi' (1)^2)}+\frac{\alpha}{\sqrt{\xi''}} L_{\Psi_*}. \qedhere
\end{align*}
\end{proof}
Recall the definition of $-E_0$ and $-y_0$ in (\ref{Ey}).
For any $\varepsilon>0$ and $\delta>0$,
we set 
\begin{align}\label{abc}
&\mathcal{A}_{N,\xi}^{\varepsilon,\delta}=\#\left\{\boldsymbol{\sigma}\in S^{N-1} \colon \nabla_{\mathrm{sp}} h_N(\boldsymbol{\sigma})=0, \; \frac{1}{\sqrt{N}} h_N(\boldsymbol{\sigma}) \in (-E_0-\varepsilon,-E_0+\varepsilon),\right.  \nonumber\\
&\quad\quad\quad\quad\quad\quad \left. \phantom{XXX\frac{1}{\sqrt{N}}} L_{N-1}\left(\nabla^2_{\mathrm{sp}}  h_N(\boldsymbol{\sigma})\right)\notin B\left(\si_{y_0,2\sqrt{\xi''}}, \delta\right)\right\},\nonumber\\ &
\mathcal{B}_{N,\xi}^{\varepsilon,\delta}=\#\left\{\boldsymbol{\sigma}\in S^{N-1} \colon \nabla_{\mathrm{sp}} h_N(\boldsymbol{\sigma})=0, \; \frac{1}{\sqrt{N}}h_N(\boldsymbol{\sigma}) \in (-E_0-\varepsilon,-E_0+\varepsilon),\right.  \nonumber\\
&\quad\quad\quad\quad\quad\quad \left. \phantom{XXX\frac{1}{\sqrt{N}}} \left|\lambda_{\min}\left(\nabla^2_{\mathrm{sp}}  h_N(\boldsymbol{\sigma})\right)-y_0+2\sqrt{\xi''}\right|>\delta\right\}\nonumber,\\ &
\mathcal{C}_{N,\xi}^{\varepsilon,\delta}=\#\left\{\boldsymbol{\sigma}\in S^{N-1} \colon \nabla_{\mathrm{sp}} h_N(\boldsymbol{\sigma})=0, \; \frac{1}{\sqrt{N}}h_N(\boldsymbol{\sigma}) \in (-E_0-\varepsilon,-E_0+\varepsilon),\right.  \nonumber\\
&\quad\quad\quad\quad\quad\quad \left. \phantom{XXX\frac{1}{\sqrt{N}}} \left|\lambda_{\max}\left(\nabla^2_{\mathrm{sp}}  h_N(\boldsymbol{\sigma})\right)-y_0-2\sqrt{\xi''}\right|>\delta\right\}.
\end{align}
\begin{lemma}\label{zeros}
If $\xi(x)$ is a pure-like mixture, then for any $\delta>0$ fixed, there exists some $\varepsilon>0$ sufficiently small such that  
\begin{align}\label{-inftys}
\lim _{N \rightarrow \infty} \mathbb{E} \mathcal{A}_{N,\xi}^{\varepsilon,\delta}=0.
\end{align}
\end{lemma}
\begin{proof}
By the Kac-Rice formula \cite[theorem 12.1.1]{AT07}, similar as to (\ref{criBD}), we obtain
\begin{align}\label{criAN}
\mathbb{E} \mathcal{A}_{N,\xi}^{\varepsilon,\delta} 
&=\omega_N \varphi_{\nabla_{\mathrm{sp}} h_N(\boldsymbol{\sigma})}(0) \mathbb{E}\left\{ \left\lvert \operatorname { d e t } ( \nabla _ { \mathrm { sp } } ^ { 2 } h_{ N } ( \boldsymbol { \sigma } ) ) \right\rvert \vphantom{1_{\frac{h_N(\boldsymbol{\sigma})}{\sqrt{N}}}} \right. \nonumber \\
&\quad \left. \times1_{\left\{\frac{h_N(\boldsymbol{\sigma})}{\sqrt{N}} \in (-E_0-\varepsilon,-E_0+\varepsilon),\frac{\partial_rh_N(\boldsymbol{\sigma})}{\sqrt{N}} \in  \mathbb{R},L_{N-1}\left(\nabla^2_{\mathrm{sp}}  h_N(\boldsymbol{\sigma})\right)\notin B\left(\si_{y_0,2\sqrt{\xi''}}, \delta\right)\right\}} \middle\vert \nabla_{\mathrm{sp}} h_N(\boldsymbol{\sigma})=0\right\}\nonumber\\&=C_{\xi,N}  \int_{-E_0-\varepsilon}^{-E_0+\varepsilon}\int_{-\8}^{\8}  \exp \left\{-\frac{N}{2}(x, y) \Sigma^{-1}(x, y)^{\top}\right\}\nonumber\\&\quad\quad\quad\times\mathbb{E}\left\{\left|\operatorname{det}\left(\GOE_{N-1}-\frac{\sqrt{N}y}{\sqrt{(N-1) \xi^{\prime \prime}}}\mathbf{I}_{N-1}\right)\right|1_{A_N^{\xi}}\right\} \dd x \dd y,    \end{align}
where the event
\begin{align*}
    A_N^{\xi} = \left\{ L_{N-1}\left(\GOE_{N-1}-\frac{\sqrt{N}y}{\sqrt{(N-1) \xi^{\prime \prime}}}\mathbf{I}_{N-1}\right)\notin B\left(\si_{\frac{\sqrt{N}y_0}{\sqrt{(N-1)\xi''}},\frac{2\sqrt{N}}{\sqrt{N-1}}}, \frac{\sqrt{N}\delta}{\sqrt{(N-1)\xi''}}\right) \right\}.
\end{align*}
Based on the result of Lemma \ref{ls}, we divide the integral (\ref{criAN}) into three parts with respect to variable $y$
\begin{align}\label{criANxi}
\mathbb{E} \mathcal{A}_{N,\xi}^{\varepsilon,\delta}&= C_{\xi,N}
 \int_{-E_0-\varepsilon}^{-E_0+\varepsilon}\int_{-y_0-\alpha\sqrt{\varepsilon}}^{-y_0+\alpha\sqrt{\varepsilon}} \exp \left\{-\frac{N}{2}(x, y) \Sigma^{-1}(x, y)^{\top}\right\} \nonumber\\
 &\quad\quad\quad\times\mathbb{E}\left\{\left|\operatorname{det}\left(\GOE_{N-1}-\frac{\sqrt{N}y}{\sqrt{(N-1) \xi^{\prime \prime}}}\mathbf{I}_{N-1}\right)\right|1_{A_N^{\xi}}\right\} \dd x \dd y \nonumber\\
 &\quad+C_{\xi,N} 
 \int_{-E_0-\varepsilon}^{-E_0+\varepsilon}\int_{-\8}^{-y_0-\alpha\sqrt{\varepsilon}}  \exp \left\{-\frac{N}{2}(x, y) \Sigma^{-1}(x, y)^{\top}\right\}\nonumber\\
&\quad\quad\quad\times\mathbb{E}\left\{\left|\operatorname{det}\left(\GOE_{N-1}-\frac{\sqrt{N}y}{\sqrt{(N-1) \xi^{\prime \prime}}}\mathbf{I}_{N-1}\right)\right|1_{A_N^{\xi}}\right\} \dd x \dd y \nonumber\\
 &\quad +C_{\xi,N} 
 \int_{-E_0-\varepsilon}^{-E_0+\varepsilon}\int_{-y_0+\alpha\sqrt{\varepsilon}}^{\8} \exp \left\{-\frac{N}{2}(x, y) \Sigma^{-1}(x, y)^{\top}\right\} \nonumber\\&\quad\quad\quad\times\mathbb{E}\left\{\left|\operatorname{det}\left(\GOE_{N-1}-\frac{\sqrt{N}y}{\sqrt{(N-1) \xi^{\prime \prime}}}\mathbf{I}_{N-1}\right)\right|1_{A_N^{\xi}}\right\} \dd x \dd y \nonumber\\
 &=:J_N^1+J_N^2+J_N^3. 
 \end{align}
 By the logarithmic asymptotics estimate (\ref{supBD}) for the mean number of critical points, removing the indicator functions in $J_N^2$ and $J_N^3$ results in
 \begin{align*}
\limsup_{N \rightarrow \infty} \frac{1}{N} \log \left(J_N^2 +J_N^3\right)\leq \sup_{\substack{x \in (-E_0-\varepsilon,-E_0+\varepsilon) \\ y \in \mathbb{R}\setminus (-y_0-\alpha\sqrt{\varepsilon},-y_0+\alpha\sqrt{\varepsilon})}} F(x, y).     
 \end{align*}
By property (\ref{smallzero}) in Lemma \ref{ls}, there exists a sufficiently small $\varepsilon>0$ such that the supremum on the right-hand side of the above inequality is strictly smaller than zero. We then directly deduce for sufficiently small $\varepsilon>0$ 
 \begin{align}\label{JN12}
\lim_{N \rightarrow \infty} (J_N^2 +J_N^3)=0.     
 \end{align}
It remains to consider $J_N^1$.
 Note that $-y_0<-2\sqrt{\xi''(1)}$ by (\ref{-y0}). The fact that $\alpha\sqrt{\varepsilon}$ tends to zero as $\varepsilon$ goes to zero allows us to use arguments similar to those in (\ref{Lip}) and (\ref{triangle}). We omit some tedious but straightforward steps to obtain that for any $\varepsilon>0$ small enough and $N$ large enough, 
 \begin{align*}
     J_N^1&\leq C_{\xi,N} 
 \int_{-E_0-\varepsilon}^{-E_0+\varepsilon}\int_{-y_0-\alpha\sqrt{\varepsilon}}^{-y_0+\alpha\sqrt{\varepsilon}}  \exp \left\{-\frac{N}{2}(x, y) \Sigma^{-1}(x, y)^{\top}\right\}\nonumber\\&\quad\times\mathbb{E}\left\{\left|\operatorname{det}\left(\GOE_{N-1}-\frac{\sqrt{N}y}{\sqrt{(N-1) \xi^{\prime \prime}}}\mathbf{I}_{N-1}\right)\right|1_{\left\{L_{N-1}\left(\GOE_{N-1}\right)\notin B\left(\si_{\rm{sc}}, \frac{\delta}{3\sqrt{\xi''}}\right)\right\}}\right\} \dd x \dd y. 
 \end{align*}
 Using the Cauchy–Schwarz inequality and replacing the factor $p(p-1)$ by $\xi''$ in estimates (\ref{second}) and (\ref{devia}), there exists a constant $C>0$ such that for any $\varepsilon>0$ sufficiently small and $N$ large enough
  \begin{align*}
     J_N^1&\leq C_{\xi,N} 
 \int_{-E_0-\varepsilon}^{-E_0+\varepsilon}\int_{-y_0-\alpha\sqrt{\varepsilon}}^{-y_0+\alpha\sqrt{\varepsilon}}\exp\left\{-CN^2\right\}  \exp \left\{-\frac{N}{2}(x, y) \Sigma^{-1}(x, y)^{\top}\right\}\nonumber\\&\quad\times\mathbb{E}\left\{\left|\operatorname{det}\left(\GOE_{N-1}-\frac{\sqrt{N}y}{\sqrt{(N-1) \xi^{\prime \prime}}}\mathbf{I}_{N-1}\right)\right|\right\} \dd x \dd y. 
 \end{align*}
 From here, we apply (\ref{supBD}) together with the property (\ref{smallvar}) in Lemma \ref{ls} to get 
 \begin{align*}
\lim _{N \rightarrow \infty} \frac{1}{N} \log J_N^1=-\infty,
\end{align*}
which completes the proof of this lemma.
 \end{proof}
\begin{lemma}\label{zero1s}
If $\xi(x)$ is a $pure$-$like$ mixture, then for any $\delta>0$ fixed, there exists some $\varepsilon>0$ sufficiently small such that
\begin{align}\label{-etas}
&\lim _{N \rightarrow \infty}  \mathbb{E} \mathcal{B}_{N,\xi}^{\varepsilon,\delta}=0,\\&\label{-eta2}\lim _{N \rightarrow \infty}  \mathbb{E} \mathcal{C}_{N,\xi}^{\varepsilon,\delta}=0.
\end{align}
\end{lemma}
\begin{proof}
To prove this lemma,
 it suffices to prove that the equality (\ref{criAN}) converges to zero as $N$ tends to infinity with the event $A_N^{\xi}$ replaced by the event
$$
B_N^{\xi} = \left\{ \left| \lambda_{\min}\left(\GOE_{N-1}\right)-\frac{\sqrt{N}(y+y_0)}{\sqrt{(N-1) \xi''}}+\frac{2\sqrt{N}}{\sqrt{N-1}}\right|> \frac{\sqrt{N}\delta}{\sqrt{(N-1)\xi''}} \right\}.
$$
Following the same partition with respect to variable $y$ as in (\ref{criANxi}), we denote the three parts by $I_N^1$, $I_N^2$ and $I_N^3$, respectively. For sufficiently small $\varepsilon>0$,  using the same argument as in (\ref{JN12}), we have 
 \begin{align}\label{IN12}
\lim_{N \rightarrow \infty} (I_N^2 +I_N^3)=0.     
 \end{align}
For $I_N^1$, since $y \in (y_0 - \alpha \sqrt{\varepsilon}, y_0 + \alpha \sqrt{\varepsilon})$, for any fixed $\delta>0$, there exists a $\varepsilon>0$ sufficiently small such that for all $N$ sufficiently large, the event $B_N^{\xi}$ implies the event  
$$\left| \lambda_{\min}\left(\GOE_{N-1}\right)+2\right|> \frac{\delta}{2\sqrt{\xi''}}.$$ 
Then we deduce for any $\varepsilon>0$ small enough and $N$ large enough
\begin{align*}
     I_N^1&\leq C_{\xi,N} 
 \int_{-E_0-\varepsilon}^{-E_0+\varepsilon}\int_{-y_0-\alpha\sqrt{\varepsilon}}^{-y_0+\alpha\sqrt{\varepsilon}}  \exp \left\{-\frac{N}{2}(x, y) \Sigma^{-1}(x, y)^{\top}\right\}\nonumber\\&\quad\times\mathbb{E}\left\{\left|\operatorname{det}\left(\GOE_{N-1}-\frac{\sqrt{N}y}{\sqrt{(N-1) \xi^{\prime \prime}}}\mathbf{I}_{N-1}\right)\right|1_{\left\{\left| \lambda_{\min}\left(\GOE_{N-1}\right)+2\right|> \frac{\delta}{2\sqrt{\xi''}}\right\}}\right\} \dd x \dd y. 
 \end{align*}
  Combining the Cauchy–Schwarz inequality, second moment upper bound (\ref{second}) and the LDP for the smallest eigenvalues of GOE matrices (\ref{el}), the logarithmic asymptotics formula (\ref{supBD}) implies that there exists a constant $C>0$ such that for all $\varepsilon>0$ sufficiently small and $N$ large enough
\begin{align}\label{III}
  I_N^1\leq \exp\{-CN \} \exp\left\{\sup _{\substack{x \in (-E_0-\varepsilon,-E_0+\varepsilon) \\ y \in  (-y_0-\alpha\sqrt{\varepsilon},-y_0+\alpha\sqrt{\varepsilon})}} N F(x, y)+o(N)\right\} .  
  \end{align}
By the key observation (\ref{smallvar}) in Lemma \ref{ls}, for sufficiently small $\varepsilon>0$, it follows that $\lim_{N \rightarrow \infty} I_N^1=0$, which together with (\ref{IN12}) finishes the proof of (\ref{-etas}). As before, the proof of (\ref{-eta2}) is a trivial extension of the method used for (\ref{-etas}), based on the LDP for the largest eigenvalues of GOE matrices. 
\end{proof}
Finally, we prove Theorem \ref{t2} with the help of Lemmas \ref{zeros} and \ref{zero1s}.
\begin{proof}[Proof of Theorem \ref{t2}] Note that $GS=-E_0$ as given in (\ref{gsm}) for the $pure$-$like$ mixture model under Assumption \ref{assum}. The proof of this theorem follows similar arguments as those in Theorem \ref{t1}. We only need to replace the parameters $c_p$ by $y_0$ and $r_p$ by $2\sqrt{\xi''}$, since  Lemmas \ref{zeros} and \ref{zero1s} can serve the same role in the proof of Theorem \ref{t2} as Lemmas \ref{zero} and \ref{zero1} do in the proof of Theorem \ref{t1}, respectively.
 \end{proof}
\section{$1$-RSB model} \label{sec4}
We prove Theorem \ref{t3} for the $1$-RSB model in this section. To give the definition of $1$-RSB model, we recall some basic facts regarding the Parisi functional of the ground state energy. 
Denote by $\mathcal{N}$ the collection of all nonnegative nondecreasing and right-continuous functions on $[0,1)$. Let
$$
\mathcal{K}=\left\{(L, \alpha) \in(0, \infty) \times \mathcal{N}: L>\int_0^1 \alpha(s) d s\right\}.
$$
For any $(L, \alpha) \in \mathcal{K}$, with $\widehat{\al}(q)=L-\int_0^q \al(s) \mathrm{d} s$, we define
\begin{align}\label{ground]}
\mathcal{Q}(L, \alpha)=\frac{1}{2}\left(\xi^{\prime}(1) L-\int_0^1 \xi^{\prime \prime}(q)\left(\int_0^q \alpha(s) d s\right) d q+\int_0^1 \frac{d q}{\widehat{\al}(q) }\right).
\end{align}
The Parisi formula for the limiting ground state energy states that \cite[Theorem 1]{chen}
\begin{align}\label{GSQ}
G S=-\inf _{(L, \alpha) \in \mathcal{K}} \mathcal{Q}(L, \alpha).
\end{align}
Since the functional $\mathcal{Q}$ is strictly convex, there exists a unique minimizer $(L_0, \alpha_0) \in \mathcal{K}$ that optimizers $\mathcal{Q}$. Additionally, the minimizer $(L_0, \alpha_0)$ admits a characterization. Set
\begin{equation}
g(q)=\xi^{\prime}(q)-\int_0^q \frac{\mathrm{~d} s}{\widehat{\alpha}(s)^2}, \quad \quad G(s)=\int_s^1 g(q) \mathrm{d} q.
\end{equation}
Let $\nu_0([0,q])$ be the measure on $[0,1]$ induced by $\alpha_0$, i.e.,
$$
\nu_{0}([0, q])=\alpha_0(q), \quad \forall q \in[0,1],
$$
which is called the Parisi measure.
We then define the set
$$
T=\{q \in[0,1]: G(q)=0\}.
$$
Then $(L_0, \alpha_0) \in \mathcal{K}$ is the minimizer of (\ref{GSQ}) \cite[Theorem 2]{chen} if and only if 
\begin{align}\label{char}
g(1)=0 ; \quad \min _{q \in[0,1]} G(q)=0 ; \quad \nu_{0}\left(T^c\right)=0.
\end{align}
\begin{definition}\label{defi}
The model $\xi$ is called $1$-RSB if $T=\{0,1\}$.  
\end{definition}
 For this kind of model, we have $\al_0=u$ for some constant $u$. Let $z>0$ be the unique solution to
$$\frac{1+z}{z^2} \log (1+z)-\frac{1}{z}=\frac{\xi}{\xi'}.$$
Set $y=\sqrt{(1+z) \xi^{\prime}(1)}$. It follows from \cite[Lemma 3.6]{huang2023constructive} that
\begin{align}\label{yy}
    y\geq \sqrt{\xi''}.
\end{align}
If $\xi$ is $1$-RSB, it was shown in\cite[Lemma 3.5]{huang2023constructive} that
$$
L_0=\frac{1+z}{y}, \quad u=\frac{z}{y}, \quad GS=-\frac{\xi^{\prime}(1)+z \xi(1)}{y},
$$
and for all $q \in[0,1]$,
$$
\xi(1)-\xi(q) \geq \xi^{\prime}(1)\left(\frac{1+z}{z^2} \log (1+(1-q) z)-\frac{1-q}{z}\right),
$$
with equality at exactly $q=0,1$. Recall the complexity function $F(x,y)$ defined in (\ref{fxy}). Set
\begin{align}\label{x0y0}
x_*=-GS=\frac{\xi^{\prime}(1)+z \xi(1)}{y}, \quad y_*=y+\frac{\xi^{\prime \prime}(1)}{y}.
\end{align} 
\begin{proof}[Proof of Theorem \ref{t3}]
Recall the notations defined in (\ref{abc}). Here we use these notations with $E_0$ replaced by $x_*$ and $y_0$ replaced by $y_*$. Moreover, we further impose the constraint $\lambda_{\min}(\nabla^2_{\mathrm{sp}}  h_N(\boldsymbol{\sigma}))\geq 0$ to these notations, since the global minimum must be a local minimum.
The proof of this theorem can follow arguments similar to those in Theorem \ref{t2}, provided we can show that for any $\delta>0$ fixed, there exists a sufficiently small $\varepsilon>0$ such that 
$\lim _{N \rightarrow \infty} \mathbb{E} \mathcal{A}_{N,\xi}^{\varepsilon,\delta}=0$, $\lim _{N \rightarrow \infty} \mathbb{E} \mathcal{B}_{N,\xi}^{\varepsilon,\delta}=0$ and
$\lim _{N \rightarrow \infty} \mathbb{E} \mathcal{C}_{N,\xi}^{\varepsilon,\delta}=0$, as stated in Lemmas \ref{zeros} and \ref{zero1s} for the $pure$-$like$ mixture model. We only provide a brief proof for $\lim _{N \rightarrow \infty} \mathbb{E} \mathcal{B}_{N,\xi}^{\varepsilon,\delta}=0$, which concerns the convergence of smallest eigenvalue. That is, we need to calculate
\begin{align}\label{criBN1}
\mathbb{E} \mathcal{B}_{N,\xi}^{\varepsilon,\delta} 
&=C_{\xi,N}  \int_{-x_*-\varepsilon}^{-x_*+\varepsilon}\int_{-\8}^{\8}  \exp \left\{-\frac{N}{2}(x, y) \Sigma^{-1}(x, y)^{\top}\right\}\nonumber\\&\quad\quad\quad\times\mathbb{E}\left\{\left|\operatorname{det}\left(\GOE_{N-1}-\frac{\sqrt{N}y}{\sqrt{(N-1) \xi^{\prime \prime}}}\mathbf{I}_{N-1}\right)\right|1_{B_N^{\xi}}\right\} \dd x \dd y,    \end{align}
where the event 
$$
B_N^{\xi} = \left\{ \left| \lambda_{\min}\left(\GOE_{N-1}\right)-\frac{\sqrt{N}(y+y_*)}{\sqrt{(N-1) \xi''}}+\frac{2\sqrt{N}}{\sqrt{N-1}}\right|> \frac{\sqrt{N}\delta}{\sqrt{(N-1)\xi''}}\quad \textbf{and}\quad \lambda_{\min}\left(\GOE_{N-1}\right)\geq \frac{\sqrt{N}y}{\sqrt{(N-1) \xi^{\prime \prime}}} \right\}.
$$
By \cite[Lemma 2.1]{belius2022complexity} and symmetry property, the complexity of local minima with radial derivative strictly larger than $-2\sqrt{\xi''}$ is strictly negative. Therefore, we only need to consider the domain of integration of the variable $y$ over $(-\infty, -2\sqrt{\xi''}+ \varepsilon)$ for any $\varepsilon>0$ and denote by $K_N(\varepsilon)$ for the integral over this domain. The assumption $y_*>2\sqrt{\xi''}$ guarantees the application of \cite[Lemma 6.4]{Be22}. Together with the LDP for the smallest eigenvalues of GOE matrices, similar to (\ref{III}), we find that there exists a constant $C>0$ independent of $\varepsilon$, such that
for all $\varepsilon>0$ sufficiently small and $N$ large enough
\begin{align}\label{II}
  K_N(\varepsilon)\leq \exp\{-CN \} \exp\left\{\sup _{\substack{x \in (-x_*-\varepsilon,-x_*+\varepsilon) \\ y \in  (-\infty,-2\sqrt{\xi''}+\varepsilon)}} N F(x, y)+o(N)\right\} .  
  \end{align}
Note that $F(x,y)=F(-x,-y)\geq F(-x,y)=F(x,-y)$ for any $x\geq 0$ and $y\geq 0$.
It was calculated in \cite[Lemma 3.18]{huang2023constructive} that
$F(-x_*,-y_*)=0.$ Furthermore, the function $F(-x_*,y)$ is strictly concave on $(-\infty,-2\sqrt{\xi''})$. Therefore, $F(-x_*,y)$ attains its maximum value zero at the unique maximizer $-y_*$ for $y\leq -2\sqrt{\xi''}$, as shown in \cite[Lemma 5.2]{huang2023constructive}. Finally, by the Lipschitz continuity of $F(x,y)$, we conclude that for any $\varepsilon>0$ small enough $\lim_{N\rightarrow\infty} K_N(\varepsilon)=0$, which finishes the proof.
\end{proof}
\begin{remark}
For the $pure$-$like$ model, we can also apply \cite[Lemma 2.1]{belius2022complexity} to reduce the integration domain of variable $y$ to $(-\infty,-2\sqrt{\xi''}+\varepsilon)$ for any $\varepsilon>0$. However, this does not overcome the essential obstacle, since the $\argmax$ function may not be continuous when $x$ lies in a small neighborhood of $-x_*$.
\end{remark}

\begin{ackn}
The authors would like to thank Qiang Zeng for introducing them to the problem of Hessian spectrum at the global minimum of the spherical mixed $p$-spin models. We are grateful to the referees for their careful reading and many constructive suggestions, especially for bringing the reference \cite{huang2023constructive} to our attention, which has significantly improved the quality of this paper.
\end{ackn}
\bibliographystyle{imsart-number}
\bibliography{gfic}
\end{document}